\documentclass[journal,10pt]{IEEEtran}
\usepackage{blindtext}
\usepackage{graphicx}
\usepackage{amsmath}
\usepackage{graphicx}
\usepackage{amsthm}
\usepackage{multirow}
\usepackage{amsmath}

\DeclareMathOperator*{\argmin}{argmin}
\usepackage{yhmath}
\usepackage{amssymb}
\usepackage{array}
\usepackage{leftidx}
\usepackage{longtable}
\usepackage{stfloats}
\usepackage{cite}
\usepackage{url}
\usepackage{subfigure}
\usepackage{comment}
\usepackage{color}
\usepackage{algorithm}
\usepackage{algorithmic}
\usepackage{lipsum}



\begin{document}
\theoremstyle{plain}
\newtheorem{thm}{Theorem}
\newtheorem{remark}{Remark}
\newtheorem{lemma}{Lemma}
\newtheorem{prop}{Proposition}
\newtheorem*{cor}{Corollary}
\theoremstyle{definition}
\newtheorem{defn}{Definition}
\newtheorem{condi}{Condition}
\newtheorem{assump}{Assumption}

\title{Potential Game-Based Decision-Making  for  Autonomous Driving}
\author{Mushuang Liu\IEEEauthorrefmark{1},~\IEEEmembership{Member,~IEEE,} ~Ilya Kolmanovsky\IEEEauthorrefmark{2},~\IEEEmembership{Fellow,~IEEE,} ~H. Eric Tseng\IEEEauthorrefmark{3}, ~Suzhou Huang\IEEEauthorrefmark{3}, ~Dimitar Filev\IEEEauthorrefmark{3},~\IEEEmembership{Fellow,~IEEE,} and Anouck Girard\IEEEauthorrefmark{2},~\IEEEmembership{Senior Member,~IEEE}

\thanks{\IEEEauthorrefmark{1} M. Liu is with the Department of Mechanical and Aerospace Engineering, University of Missouri, Columbia, MO, USA (email: ml529@missouri.edu).
}
\thanks{\IEEEauthorrefmark{2} I. Kolmanovsky and A. Girard,  are with the Department of  Aerospace Engineering, University of Michigan, Ann Arbor, MI, USA (email: ilya@umich.edu and anouck@umich.edu).
}
\thanks{\IEEEauthorrefmark{3} H. E. Tseng, S. Huang,  and D. Filev are with Ford Research and Innovation Center, 2101 Village Road, Dearborn, MI 48124, USA (e-mail: htseng@ford.com, shuang10@ford.com, and dfilev@ford.com).}
\thanks{This work is supported by Ford Motor Company.}
}
\maketitle
\markboth{IEEE Transactions on Intelligent Transportation Systems} {Liu \MakeLowercase{\textit{et al.}}: Potential Game Based  Frameworks for Decision-Making in Autonomous Driving}
\begin{abstract}
Decision-making for autonomous driving is challenging, considering the complex interactions among multiple traffic agents (including autonomous vehicles (AVs), human-driven vehicles, and pedestrians) and the computational load needed to evaluate these interactions. 
This paper develops two  general potential game-based  frameworks, namely, finite and continuous potential games, for decision-making in autonomous driving.  The two frameworks account for the AVs' two types of action spaces, i.e., finite and continuous action spaces, respectively. The developed frameworks provide theoretical guarantees for the existence of pure-strategy Nash equilibria and for the convergence of the Nash equilibrium (NE) seeking algorithms. The scalability challenge is also addressed. In addition, we provide cost function shaping approaches such that the agents' cost functions not only reflect common driving objectives but also yield potential games. The performance of the developed algorithms is demonstrated in diverse traffic scenarios, including  intersection-crossing and lane-changing scenarios.  Statistical comparative studies, including 1) finite potential game vs. continuous potential game, 2) best response dynamics vs. potential function optimization, {and 3) potential game vs. reinforcement learning (RL) vs. control barrier function (CBF)}, are conducted to compare the robustness  against various surrounding vehicles' strategies and to compare the computational efficiency.  
 {It is shown that the developed potential game frameworks have better robustness than RL and than CBF if the surrounding vehicles are not safety-conscious, and are computationally feasible for real-time implementation.} 
\end{abstract}

\section{Introduction}
 During the past five years, the autonomous vehicle (AV) market has attracted more than $\$50$ billion {in investment} from major carmakers, tech giants and start-ups, and is expected to continue its rapid growth \cite{online1}. Although the benefits of AVs are substantial,  technical challenges, such as  intelligent decision-making  in diverse  traffic scenarios, still remain to be addressed  before the AVs can {routinely} drive on the roads \cite{challenge3,challenge1,challenge2}. Developing   practical and reliable decision-making algorithms for AVs is challenging, due to the complexities of the interactions of multiple traffic participants  and the requirement that decisions are made in real-time \cite{decision_add1,decision2,decision3}.

To generate  safe and effective decisions for the ego vehicle, the interactions with its surrounding traffic agents, including AVs, human-driven vehicles, and  pedestrians, have to be considered\cite{interaction_2,interaction_1,interaction_3,nan_game}. As  each traffic agent  has its own objective (also called self-interest), the multi-agent decision-making problem is intrinsically a multi-player game problem \cite{decision3,game_sta,Victor,mine_1,game_new_1}. In a game-theoretic {setting},  each agent aims to optimize its self-interest, which is affected by not only its own actions but also the actions of its surrounding agents.   {Along these lines, reference}  \cite{game_merge} develops a two-player nonzero-sum game to model  the interactions between a highway vehicle and a merging vehicle in merging-giveaway scenarios.  Papers \cite{game_sta,Victor} {consider} a two-player Stackelberg game and a pairwise normal-form game, respectively,  to address the decision-making in lane-changing scenarios.
  To capture the agents' interactions  in  intersection-crossing scenarios, a leader-follower game and a normal-form game are introduced, respectively, in  \cite{game_leader} and \cite{mine_1}.  To solve for Nash equilibrium (NE) {in} multi-agent Markov  games,  a best response dynamics based solution algorithm is proposed in   \cite{suzhou}, with numerical examples in  two-vehicle highway merging scenarios. {Although these game-theoretic approaches have shown effectiveness in characterizing agents' interactions and in generating human-like negotiating behaviors,  technical challenges, including existence of pure-strategy solutions, convergence of solution seeking algorithms, and computational scalability, still remain to be addressed before the game-theoretic approaches can be widely-accepted as a practical and reliable solution to autonomous driving.    
  This paper aims to address these  challenges by developing a novel potential game based framework integrated with receding horizon optimization.}
  
  
 {Potential games \cite{first} are a special class of multi-player games where a real-valued global function (called potential function) exists such that a change of an agent's self-interest by its own strategy deviation equals the change of the potential function. Potential games have appealing properties, including existence and attainability of pure-strategy Nash equilibria, and as such, have attracted increasing attention in economics \cite{economics}, social systems \cite{nature}, and wireless networks \cite{wirless}. However, verifying whether a given game is a potential game is not easy \cite{hard}, limiting  broader applications of potential games to engineering problems such as autonomous driving.}

{This paper proposes a novel formulation of potential games to address the decision-making in autonomous driving.  
The contributions are fourfold:
\begin{enumerate}
    \item A finite potential game-based framework is proposed, where the AV action space {contains} a finite number of elements. This framework features an integration of receding horizon optimization and potential games, and a cost function shaping approach such that the designed cost functions not only reflect common driving objectives but also yield potential games.  Existence and convergence to pure-strategy NE  is guaranteed, and the scalability challenge is also addressed. 
    
    \item A continuous potential game-based framework is developed, where the AV action space {contains} an infinite number of elements. Theoretical analysis, cost function shaping approaches, and scalable NE seeking algorithms are provided. 
    \item Practical formulations of the proposed potential games for specific traffic scenarios, including  intersection-crossing and lane-changing scenarios, are provided. Numerical studies in these scenarios are conducted.
    
    \item Statistical comparisons are performed for  a) finite potential game vs.  continuous potential game, b) best response dynamics vs. potential function optimization, and c) potential game vs. reinforcement learning vs. control barrier function based approaches. These comparisons highlight the  robustness and computational efficiency characteristics of various algorithms. 
\end{enumerate}
}


This paper is organized as follows.  Section \ref{II} formulates the AV decision-making as a multi-player game problem.   Section  \ref{sec:finite} develops a finite potential game framework for AVs {in a setting of} finite action spaces. Section \ref{sec_continuous} develops a continuous potential game framework for AVs {in a setting of}  continuous action spaces.  Section \ref{simulation}  conducts numerical studies,  and Section \ref{VI} concludes the paper.

\section{Problem formulation}\label{II}
Consider a group of traffic agents $\mathcal{N}=\{1,2,...,N\}$ sharing the road. The dynamics of  each agent are described by the following {discrete-time equations}.
\begin{equation}\label{dynamics}
    x_i(t+1)=f_i(x_i(t),u_i(t)),
\end{equation}
where $i\in\mathcal{N}$, $x_i(t)\in\mathcal{X}_i$ and $u_i(t)\in\mathcal{U}_i$ are, respectively, the state and action of agent $i$ at the time step $t$,   $\mathcal{X}_i$ and $\mathcal{U}_i$ are the state space and  action space of agent $i$, respectively,  and $f_i$ is the system evolution model of agent $i$. We denote the dimension of $x_i$ and $u_i$ as $n_i$ and $m_i$, respectively.  In autonomous driving applications, $x_i$ usually contains the vehicle's  position and velocity information, and $a_i$ represents maneuvers such as braking and/or lane-changing. 

In a driving scenario, every agent  has its own objective, e.g., tracking its desired speed while avoiding collisions. We denote the performance index of agent $i$ at time step $t$ as $J_i(x_i(t),u_i(t),x_{-i}(t),u_{-i}(t))$, where $x_{-i}$ and $u_{-i}$ represent the states and actions of all other agents except for agent $i$, i.e., $x_{-i}=\{x_1,x_2,...,x_{i-1},x_{i+1},...,x_N\}$, $u_{-i}=\{u_1,u_2,...,u_{i-1},u_{i+1},...,u_N\}$.  {Define  $x(t)$  as the global system state, $x(t)=\{x_i(t),x_{-i}(t)\}$, and $u(t)$ as the global action $u(t)=\{u_i(t),u_{-i}(t)\}$.} At each time step $t$, an agent $i$ plans its optimal action sequence (also referred to as strategy)  for the next $T$ ($T\geq 1$) time periods, i.e., $\mathbf{u}^*_i(t)=\{u^*_i(t),u^*_i(t+1),...,u^*_i(t+T-1)\}$, to minimize the cumulative cost over the prediction horizon $\mathcal{T}=\{t,t+1,...,t+T-1\}$, that is,
\begin{equation}\label{value}
\begin{split}
    \mathbf{u}^*_i(t)&\in\argmin_{\mathbf{u}_i(t)\in\mathcal{S}_i} V_i^t(\mathbf{u}_i(t),\mathbf{u}_{-i}(t))\\
    &=\argmin_{\mathbf{u}_i(t)\in\mathcal{S}_i} \sum_{\tau=t}^{t+T-1}J_i(x_i(\tau),u_i(\tau),x_{-i}(\tau),u_{-i}(\tau)),
    \end{split}
\end{equation}
where $\mathcal{S}_i$ is the strategy space of agent $i$ and is determined by its action space $\mathcal{U}_i$. Denote $\mathcal{S}=\mathcal{S}_1\times\mathcal{S}_2\times...\times\mathcal{S}_N$ and $\mathbf{u}(t)=\{\mathbf{u}_i(t),\mathbf{u}_{-i}(t)\}$, where $\mathbf{u}_{-i}=\{\mathbf{u}_1,\mathbf{u}_2,...,\mathbf{u}_{i-1},\mathbf{u}_{i+1},...,\mathbf{u}_N\}$. Here $V_i^t:\mathcal{S}\rightarrow \mathbb{R}$ is the cumulative cost over the prediction horizon.  
It is assumed that the initial states of all agents are available, i.e., $x(t)$ is known to all $i\in\mathcal{N}$ at time $t$.

To enable an AV to respond timely to any changes in its environment (e.g., a sudden brake of the front vehicle), we employ the receding-horizon optimization: At time $t$, an AV plans its optimal action sequence  $\mathbf{u}^*_i(t)$ according to \eqref{value}, implements the first element $u^*_i(t)$ only, and then repeats the process at the next time step $t+1$ with a shifted horizon. 

As shown in \eqref{value}, agent $i$'s cost is affected by not only its own strategy $\mathbf{u}_i(t)$, but also the strategies of its surrounding agents $\mathbf{u}_{-i}(t)$. As such, the multi-agent optimization problem  \eqref{value} is intrinsically a multi-player game problem. 

We denote the game at $t$ as $\mathcal{G}^t=\{\mathcal{N},\mathcal{S},\{V_i^t\}_{i\in\mathcal{N}}\}$, which is the receding horizon multi-player game we aim to solve. The set of all agents' optimal strategies $\{\mathbf{u}_1^*(t),\mathbf{u}_2^*(t),...,\mathbf{u}_N^*(t)\}$ that satisfy \eqref{value}, if nonempty, yields a pure-strategy Nash equilibrium defined as follows. 

\begin{defn} [Pure-Strategy Nash Equilibrium \cite{game_book}] \label{d2}
An $N$-tuple of strategies $\{\mathbf{u}_1^*(t),\mathbf{u}_2^*(t),...,\mathbf{u}_N^*(t)\}$ is  a pure-strategy Nash equilibrium for the $N$-player game \eqref{value} if and only if   
\begin{equation}\label{Nash_de}
\begin{split}
& V_i^t(\mathbf{u}_i^*(t),\mathbf{u}_{-i}^*(t))\leq V_i^t(\mathbf{u}_i(t),\mathbf{u}_{-i}^*(t)),\\
 &\qquad \forall i\in\mathcal{N}, \forall \mathbf{u}_i(t)\in\mathcal{S}_i.
 \end{split}
\end{equation}
\end{defn}
Equation \eqref{Nash_de} implies that if a pure-strategy NE is achieved, then  no player would have the incentive to change its strategy unilaterally, and the system is in an equilibrium state. Since we are considering autonomous driving applications and the strategy usually reflects the AV maneuvers, we only focus on  pure-strategy NE in this paper, and the word ``pure-strategy" may be omitted if no confusion. 

{
Although the receding horizon multi-player game \eqref{value} can model the AV decision-making nicely, solving such a multi-player game is not straightforward in general. The challenges include 1) Given arbitrary $V_i^t$, the NE $\mathbf{u}^*(t)$ that satisfies \eqref{Nash_de} may not always exist \cite{add_PSNE}; 2) Even if there exists a NE, the NE seeking algorithm, e.g., the best-response dynamics, may not always converge \cite{first}; 3) Even if the algorithm converges, solving for NE is computationally expensive in general, and the computational complexity can exponentially increase with the number of agents \cite{add_exponential}. These challenges limit the practical applicability of  multi-player games to autonomous driving. We aim to address these challenges in this paper by formulating the game \eqref{value} as a potential game  by  designing the AV cost functions, $V_i^t$, appropriately for various driving scenarios so that they can reflect common driving objectives while leading to potential games. }

\section{Finite potential game}\label{sec:finite}
In this section, we consider the multi-player game \eqref{value} in the setting of finite strategy spaces. Specifically, $\mathcal{S}_i$ has a finite number of elements for all $i\in\mathcal{N}$. 
{\subsection{Preliminaries}\label{finite_def}}
We first recall the definition of finite potential games. 
\begin{defn} [Finite Exact Potential Game \cite{potential_book}] \label{d3}
The game \eqref{value} is a finite exact potential game if and only if the strategy space $\mathcal{S}_i$ ($ i\in\mathcal{N}$) contains a finite number of elements  and a potential function {$F^t:\mathcal{S}\rightarrow \mathbb{R}$} exists such that, $\forall i\in\mathcal{N}$,
\begin{equation}\label{Potential_def}
\begin{split}
& V_i^t(\mathbf{u}_i(t),\mathbf{u}_{-i}(t))- V_i^t(\mathbf{u}'_i(t),\mathbf{u}_{-i}(t))\\
 &= F^t(\mathbf{u}_i(t),\mathbf{u}_{-i}(t))- F^t(\mathbf{u}'_i(t),\mathbf{u}_{-i}(t)),\\
& \forall \mathbf{u}_i(t),\mathbf{u}'_i(t)\in \mathcal{S}_i, \text{and }   \forall \mathbf{u}_{-i}(t)\in \mathcal{S}_{-i}.
\end{split}
\end{equation}
\end{defn}
Equation \eqref{Potential_def} implies that the change of agent $i$'s {cumulative cost} caused by its own strategy deviation leads to exactly the same amount of change in the potential function. Note that the potential function $F^t$  does not have a subscript, implying that it is {the same} for all agents. 
This paper only considers exact potential games, and the word  ``exact" will subsequently be omitted. 

{Potential games have appealing properties. Two of these properties are summarized by the following two lemmas, which will be used to facilitate the analysis in this paper.}

\begin{lemma}\label{l1}[Existence of pure-strategy Nash equilibria \cite{first}] If the game \eqref{value} is a finite exact potential game, then there exists  at least one pure-strategy Nash equilibrium.
\end{lemma}


\begin{lemma}\label{l2}[Equivalence of Nash equilibrium sets \cite{potential_book}]
If the game \eqref{value} is a finite exact potential game with $F^t$ as a potential function, then the set of pure-strategy Nash equilibria of \eqref{value} coincides with the set of pure-strategy Nash equilibria of the identical-interest game with all agents' cumulative cost equal to the potential function. That is,
\begin{equation}
    \text{NESet}(\mathcal{G}^t)=\text{NESet}(\Tilde{\mathcal{G}}^t),
\end{equation}
where $\text{NESet}$ denotes the set of pure-strategy NE, and $\Tilde{\mathcal{G}}^t=\{\mathcal{N},\mathcal{A},\{F^t\}_{i\in\mathcal{N}}\}$.
\end{lemma}
{\subsection{Constructing finite potential games for autonomous driving}\label{finite_fomulation}}
   
Consider the decision-making in autonomous driving according to \eqref{value}. In this subsection, we show how to design $V_i^t$ such that \eqref{value} is a finite potential game. 

Theorem \ref{t1} shows that if $V_i^t$ depends only on $\mathbf{u}_i(t)$, then the resulting game is a finite potential game. 

\begin{thm}\label{t1}
Let the cumulative cost in \eqref{value} be of the following form,
\begin{equation}\label{finite_self}
    V_i^t(\mathbf{u}_i(t),\mathbf{u}_{-i}(t))=V_i^{t,self}(\mathbf{u}_i(t)),
\end{equation}
 where $V_i^{t,self}$ is a function that is {determined} solely by agent $i$'s strategy $\mathbf{u}_i(t)$. Then the game \eqref{value} is a finite exact potential game with a potential function $P^t$ defined as
\begin{equation}\label{p_1}
    P^t(\mathbf{u}(t))=\sum_{i\in\mathcal{N}}V_i^{t,self}(\mathbf{u}_i(t)).
\end{equation}
\end{thm}
\begin{proof}
Given the cumulative cost \eqref{finite_self}, for any unilateral strategy deviation of agent $i$ from $\mathbf{u}_i$ to $\mathbf{u}'_i$, the following equation holds,
\begin{equation}
\begin{split}
    &V_i^t(\mathbf{u}_i(t),\mathbf{u}_{-i}(t))-V_i^t(\mathbf{u}'_i(t),\mathbf{u}_{-i}(t))\\
    &=V_i^{t,self}(\mathbf{u}_i(t))-V_i^{t,self}(\mathbf{u}'_i(t)).
    \end{split}
\end{equation}

Consider the function $P^t$  in \eqref{p_1}. The strategy deviation from $\mathbf{u}_i(t)$ to $\mathbf{u}'_i(t)$ leads to the change of  this function as
\begin{equation}
\begin{split}
   & P^t(\mathbf{u}_i(t),\mathbf{u}_{-i}(t))-P^t(\mathbf{u}'_i(t),\mathbf{u}_{-i}(t))\\&=V_i^{t,self}(\mathbf{u}_i(t))+\sum_{j\in\mathcal{N},j\neq i}V_j^{t,self}(\mathbf{u}_j(t))\\
   &\quad-V_i^{t,self}(\mathbf{u}'_i(t))-\sum_{j\in\mathcal{N},j\neq i}V_j^{t,self}(\mathbf{u}_j(t))\\
   &=V_i^{t,self}(\mathbf{u}_i(t))-V_i^{t,self}(\mathbf{u}'_i(t)).
\end{split}
\end{equation}
Because $V_i^t(\mathbf{u}_i(t),\mathbf{u}_{-i}(t))-V_i^t(\mathbf{u}'_i(t),\mathbf{u}_{-i}(t))=P^t(\mathbf{u}_i(t),\mathbf{u}_{-i}(t))-P^t(\mathbf{u}'_i(t),\mathbf{u}_{-i}(t))$ holds for all $i\in\mathcal{N}$, $\mathbf{u}_i(t),\mathbf{u}'_i(t)\in \mathcal{S}_i, \text{ and }  \mathbf{u}_{-i}(t)\in \mathcal{S}_{-i}$, the function $P^t$ is a potential function, and the formulated  game is a finite potential game.
\end{proof}

{
\begin{remark}
Theorem \ref{t1} states that if the cumulative cost $V_i^t$ satisfies \eqref{finite_self}, then the game \eqref{value} is a finite potential game. Such a cost can characterize self-centered driving objectives, such as ride comfort (often reflected by vehicle jerk \cite{add_jerk}), travel efficiency (captured by vehicle speed \cite{add_travelefficiency}), fuel efficiency (dependent on speed, acceleration, and jerk \cite{add_fuel}), keeping in the center of a lane (reflected by vehicle lateral position), and so on. A specific cost function that motivates the AV to track a desired speed and that satisfies \eqref{finite_self} is provided in the simulation section by \eqref{trackspeed}.
\end{remark}}

Theorem \ref{t2} considers  the cases where an AV's objective is jointly affected by both $\mathbf{u}_i(t)$ and $\mathbf{u}_{-i}(t)$, in a pairwise symmetric manner.
\begin{thm}\label{t2}
Let the cumulative cost in \eqref{value} be of the following form{,}
\begin{equation}\label{finite_pair}
    V_i^t(\mathbf{u}_i(t),\mathbf{u}_{-i}(t))=\sum_{j\in\mathcal{N},j\neq i} V_{ij}^t(\mathbf{u}_i(t),\mathbf{u}_{j}(t)),
\end{equation}
where $V_{ij}^t(\mathbf{u}_i(t),\mathbf{u}_{j}(t))=V_{ji}^t(\mathbf{u}_j(t),\mathbf{u}_{i}(t)), \forall i,j\in\mathcal{N}, i\neq j, \forall \mathbf{u}_i(t)\in \mathcal{S}_i,\forall\mathbf{u}_{j}(t)\in \mathcal{S}_j$.
Then the game \eqref{value} is a finite exact potential game with the following potential function{,} 
\begin{equation}\label{p_2}
    G^t(\mathbf{u}(t))=\sum_{i\in \mathcal{N}}\sum_{j\in\mathcal{N},j< i} V_{ij}^t(\mathbf{u}_i(t),\mathbf{u}_{j}(t)).
\end{equation}
\end{thm}
\begin{proof}
The strategy deviation of  agent $i$ from $\mathbf{u}_i$ to $\mathbf{u}'_i$ leads to the change  in its cumulative cost.
\begin{equation}
\begin{split}
    &V_i^t(\mathbf{u}_i(t),\mathbf{u}_{-i}(t))-V_i^t(\mathbf{u}'_i(t),\mathbf{u}_{-i}(t))\\
    &=\sum_{j\in\mathcal{N},j\neq i} V_{ij}^t(\mathbf{u}_i(t),\mathbf{u}_{j}(t))-\sum_{j\in\mathcal{N},j\neq i} V_{ij}^t(\mathbf{u}'_i(t),\mathbf{u}_{j}(t)).
    \end{split}
\end{equation}

The   deviation from $\mathbf{u}_i(t)$ to $\mathbf{u}'_i(t)$ leads to the change in the function $G^t$ defined in \eqref{p_2} as
\begin{equation}\label{E15}
\begin{split}
   & G^t(\mathbf{u}_i(t),\mathbf{u}_{-i}(t))-G^t(\mathbf{u}'_i(t),\mathbf{u}_{-i}(t))\\
    &=\sum_{j\in\mathcal{N},j\neq i} V_{ij}^t(\mathbf{u}_i(t),\mathbf{u}_{j}(t))-\sum_{j\in\mathcal{N},j\neq i} V_{ij}^t(\mathbf{u}'_i(t),\mathbf{u}_{j}(t)).
\end{split}
\end{equation}
Equation \eqref{E15} holds because $V_{ij}^t(\mathbf{u}_i(t),\mathbf{u}_{j}(t))=V_{ji}^t(\mathbf{u}_j(t),\mathbf{u}_{i}(t))$  $\forall i,j\in\mathcal{N}, i\neq j, \forall \mathbf{u}_i(t)\in \mathcal{S}_i,  \forall  \mathbf{u}_{j}(t)\in \mathcal{S}_j$.  
As such, the function $G^t$ is a potential function, and the formulated game is a finite exact potential game.
\end{proof}

{
\begin{remark}
Theorem \ref{t2} states that if the cumulative cost $V_i^t$ satisfies \eqref{finite_pair}, then the game \eqref{value} is a finite potential game.    In a typical driving scenario, vehicles interactions are often brought about by the desire of avoiding collisions. The cost design in \eqref{finite_pair} can characterize such an interaction, by applying a symmetric collision penalty to both vehicles if a collision between them happens. A specific example is provided in Section \ref{simulation} by \eqref{collision_design} and \eqref{tanh_cost}.
\end{remark}
}

Next theorem shows that if an AV's objective function contains both  self-dependent and pairwise components, then the resulted game is still a potential game.   
\begin{thm}\label{t3}
Let the cumulative cost in \eqref{value} be of the following form,
\begin{equation}\label{design_value}
\begin{split}
    &V_i^t(\mathbf{u}_i(t),\mathbf{u}_{-i}(t))\\
    &=\alpha V_i^{t,self}(\mathbf{u}_i(t))+\beta \sum_{j\in\mathcal{N},j\neq i} V_{ij}^t(\mathbf{u}_i(t),\mathbf{u}_{j}(t)),
\end{split}
\end{equation}
where $V_i^{t,self}(\mathbf{u}_i(t))$ and $V_{ij}^t(\mathbf{u}_i(t),\mathbf{u}_{j}(t))$ satisfy \eqref{finite_self} and \eqref{finite_pair}, respectively, and $\alpha\in\mathbb{R}$ and $\beta\in\mathbb{R}$. Then the game \eqref{value} is a finite exact potential game with the following potential function, 
\begin{equation}\label{total_potential}
\begin{split}
         &F^t(\mathbf{u}(t))=\alpha P^t(\mathbf{u}(t))+\beta G^t(\mathbf{u}(t))\\
    & =\alpha\sum_{i\in\mathcal{N}}V_i^{t,self}(\mathbf{u}_i(t))+\beta\sum_{i\in \mathcal{N}}\sum_{j\in\mathcal{N},j< i} V_{ij}^t(\mathbf{u}_i(t),\mathbf{u}_{j}(t)).
\end{split}
\end{equation}
\end{thm}
\begin{proof}
The proof is straightforward by combining Theorems \ref{t1} and \ref{t2}.
\end{proof}

{
\begin{remark}
    Theorem \ref{t3} provides a general template to design AV cost functions that yield finite potential games. The cost \eqref{design_value} can reflect driving objectives of, for example, tracking a desired speed while avoiding collisions with other vehicles.  We {note} that most of the AVs' cost function models in the existing works follow, or can be slightly revised to follow, the form proposed in \eqref{design_value}. See \cite{payoff1,payoff2,payoff3,dynamics} for {examples}. 
\end{remark}
}
{\subsection{Solving the constructed finite potential games}}
After formulating the game \eqref{value} as a finite potential game,  according to Lemma \ref{l1}, a pure-strategy NE always exists. Next we show how to solve this potential game and find the NE. 

One of the most commonly-used algorithms to seek NE is the best response dynamics \cite{add_exponential}, which works as follows: Given an initial guess of all agents' strategies, we find the best response of each agent to the strategies of others, i.e., $\mathbf{u}^*_j(t)\in\argmin_{\mathbf{u}_j(t)\in\mathcal{S}_j} V_j^t(\mathbf{u}_j(t),\mathbf{u}_{-j}(t))$ for all $j\in\mathcal{N}$; After all agents' strategies are updated, we then repeat the process iteratively until no agent has the incentive to change its strategy. The detailed algorithm is shown in Algorithm \ref{A1}.

\begin{algorithm}[t]
\caption{Best response dynamics to solve \eqref{value}} \label{A1}
\hspace*{0.0in} {\bf Inputs:} \\ 
\hspace*{0.02in}Agent set: $\mathcal{N}$;  \\
\hspace*{0.02in}Global system state: $x(t)$;  \\
\hspace*{0.02in}Global strategy space: $\mathcal{S}$;  \\
\hspace*{0.02in}Cumulative cost function of each agent: $V_i^t,\forall j\in\mathcal{N}$;\\
\hspace*{0.02in}System dynamics of each agent: $f_j, \forall j\in\mathcal{N}$.\\
\hspace*{0.02in} {\bf Output:} \\
\hspace*{0.02in} The ego vehicle optimal strategy  (in the sense of NE): $\mathbf{u}^*_i(t)$.\\
\hspace*{0.02in} {\bf Procedures:} 
\begin{algorithmic}[1]
\STATE Set \textit{NashCondition=False}
\STATE {\bf While} \textit{NashCondition=False} {\bf do}
\STATE \hspace*{0.2in}{\bf For} $j=1,2,...,N$ {\bf do}
\STATE \hspace*{0.3in} Find $\mathbf{u}^*_j(t)$ according to 
\hspace*{0.3in} 
\begin{equation}\nonumber    \quad\quad\quad\mathbf{u}^*_j(t)\in\argmin_{\mathbf{u}_j(t)\in\mathcal{S}_j} V_j^t(\mathbf{u}_j(t),\mathbf{u}_{-j}(t)).
\end{equation}
\STATE \hspace*{0.3in} Update $\mathbf{u}_j(t)$ using $\mathbf{u}^*_j(t)$.
\STATE \hspace*{0.2in} {\bf End for}
\STATE \hspace*{0.2in}{\bf If} $    \mathbf{u}^*_j(t)\in\argmin\limits_{\mathbf{u}'_j(t)\in\mathcal{S}_j} V_j^t(\mathbf{u}'_j(t),\mathbf{u}^*_{-j}(t))$
\STATE  \hspace*{0.3in} holds    $\forall j\in\mathcal{N}$,
\STATE  \hspace*{0.16in} {\bf Then}  Set \textit{NashCondition=True}.
\STATE \hspace*{0.2in}{\bf End if}
\STATE {\bf End while}
\end{algorithmic}
\end{algorithm}

\begin{thm}\label{l3}[Convergence of Algorithm \ref{A1}]
    Consider the game \eqref{value} and the Algorithm \ref{A1}. If the cumulative cost $V_i^t$ in \eqref{value} is designed according to \eqref{design_value}, then Algorithm \ref{A1} converges within a finite number of iterations.
\end{thm}
\begin{proof}
    According to Theorem \eqref{t3}, if the  $V_i^t$ in \eqref{value} satisfies \eqref{design_value}, then the game   \eqref{value} is a finite potential game. Because the best response dynamics in Algorithm \ref{A1} generates an improvement path, and because every improvement path must terminate within a finite number of iteration steps  in finite potential games \cite{first}, Algorithm \ref{A1} converges to NE within a finite number of iterations.
\end{proof}

Although the convergence of the best response dynamics is guaranteed, this algorithm is generally computationally expensive, as it requires multiple and iterative optimizations at each $t$. Next we show that if the game is a potential game, then the pure-strategy NE can be derived by solving only one optimization at each $t$:
\begin{equation}\label{optimization}
    \mathbf{u}^*(t)\in\argmin_{\mathbf{u}(t)\in\mathcal{S}} F^t(\mathbf{u}(t)),
\end{equation}
where {$F^t$} is the potential function of the game \eqref{value}. The detailed algorithm that uses this potential function optimization approach to solve the game \eqref{value} is shown in Algorithm \ref{A2}.

\begin{algorithm}[t]
\caption{Potential function optimization to solve \eqref{value}} \label{A2}
\hspace*{0.0in} {\bf Input:} \\ 
\hspace*{0.02in}Agent set: $\mathcal{N}$;  \\
\hspace*{0.02in}Global system state: $x(t)$;  \\
\hspace*{0.02in}Global strategy space: $\mathcal{S}$;  \\
\hspace*{0.02in}Cumulative cost function of each agent: $V_i^t,\forall j\in\mathcal{N}$;\\
\hspace*{0.02in}System dynamics of each agent: $f_j, \forall j\in\mathcal{N}$.\\
\hspace*{0.02in} {\bf Output:} \\
\hspace*{0.02in} The ego vehicle optimal strategy (in the sense of  NE): $\mathbf{u}^*_i(t)$.\\
\hspace*{0.02in} {\bf Procedure:} 
\begin{algorithmic}[1]
\STATE Find potential function $F^t$ according to \eqref{total_potential}.
\STATE Find $\mathbf{u}^*(t)$ according to \eqref{optimization}.
\end{algorithmic}
\end{algorithm}

\begin{thm}\label{l4}
    Consider the game \eqref{value} and the Algorithm \ref{A2}. If the cumulative cost $V_i^t$ in \eqref{value} is designed according to \eqref{design_value}, then $\mathbf{u}^*(t)$ from \eqref{optimization} is a pure-strategy NE. 
\end{thm}
\begin{proof}
    According to Theorem \ref{t3}, if the  $V_i^t$ in \eqref{value} satisfies \eqref{design_value}, then the game   \eqref{value} is a finite potential game. According to Lemma \ref{l2}, if \eqref{value} is a finite potential game, then the NE can be derived by solving the identical-interest game \eqref{optimization}.
\end{proof}

{
\begin{remark}
   In autonomous driving, real-time evaluation and planning are critical for AVs to generate timely response to any changes in the environment. If the game \eqref{value} is not a potential game, then solving it generally requires multiple and iterative optimizations, as shown in Algorithm \ref{A1}; these could be too time-consuming to be employed in autonomous driving. 
   
   In addition, the solutions from Algorithms \ref{A1} and \ref{A2} may be different due to the existence of multiple NE. In this case, the one from  Algorithm \ref{A2} should be preferable from a ``social cost" perspective.  It is because the solution from Algorithm \ref{A2} is not only a NE but also a global minimizer of the ``social cost" characterized by the potential function $F^t$, which usually takes into consideration all agents' cost as shown in \eqref{total_potential}. In an autonomous driving setting, it means that if a maneuver can benefit not only the AV itself (in the sense of individual optimality, $\min V_i^t$), but also the surrounding vehicles and pedestrians (in the sense of social optimality, $\min F^t$), then it should be preferable to the ones that are only individually optimal. Indeed the AV is expected to be considerate to other agents on the road, instead of caring for its self-interest only. 
\end{remark}
}

\section{Continuous potential games }\label{sec_continuous}
In this section, we consider the multi-player game \eqref{value} {with} continuous strategy spaces  $\mathcal{S}_i$. As a continuous action space contains infinitely many elements, the finite potential game results investigated in Section \ref{sec:finite} may not always hold. Therefore, in this section, we develop results for  continuous potential games.

{\subsection{Preliminaries}}
Let us first define continuous potential games.
\begin{defn} [Continuous Exact Potential Game] \label{d5}
The game \eqref{value} 
is a continuous exact potential game if and only if $\mathcal{S}_i\subset \mathbb{R}^{m_i}$ is a connected set and a potential function $F^t:\mathcal{S}\rightarrow\mathbb{R}$ exists such that, $\forall i\in\mathcal{N}$, $\forall \mathbf{u}_i(t)\in\mathcal{S}_i,  \forall \mathbf{u}_{-i}(t)\in\mathcal{S}_{-i}$,
\begin{equation}\label{defi_c}
\begin{split}
   & \frac{\partial V_i^t(\mathbf{u}_i(t),\mathbf{u}_{-i}(t))}{\partial \mathbf{u}_i(t)}=\frac{\partial F^t(\mathbf{u}_i(t),\mathbf{u}_{-i}(t))}{\partial \mathbf{u}_i(t)}, 
\end{split}
\end{equation}
where  $V_i^t$ and $F^t$ are  everywhere differentiable  on an open superset of $\mathcal{S}$, and $\mathcal{S}_i$ is  nonempty $\forall i\in\mathcal{N}$. 

Note that the above definition is slightly different from the one in \cite{potential_book}. Specifically, $\mathbf{u}_i$ is defined as a scalar in \cite{potential_book}, while it is a vector in Definition \ref{d5}. {With this generalization, to ensure that the theoretical analysis still holds, we include the proofs of each lemma in this section in the Appendix.}
\end{defn}

Proposition \ref{p1} shows the equivalence of the finite and continuous potential games under certain conditions. 
\begin{prop}\label{p1}
Assume that  $V_i^t$ and $F^t$ are  everywhere differentiable  on an open superset of $\mathcal{S}$, and $\mathcal{S}_i$ is  nonempty $\forall i\in\mathcal{N}$. If $\mathcal{S}_i\subset \mathbb{R}^{m_i}$ is a compact (i.e., bounded and closed) and connected set  $\forall i\in\mathcal{N}$, then   \eqref{defi_c} holds if and only if  \eqref{Potential_def} holds.
\end{prop}
\begin{proof}
    See Appendix A. 
\end{proof}


\begin{lemma}\label{l4}[Equivalence of Nash equilibrium sets]
If the game \eqref{value} is a continuous exact potential game and $F^t$ is a potential function, then the set of pure-strategy Nash equilibria of \eqref{value} coincides with the set of pure-strategy Nash equilibria for the identical-interest game with all agents' cumulative cost equal to the potential function $F^t$.
\end{lemma}
\begin{proof}
    See Appendix B. 
\end{proof}

\begin{lemma}\label{l5}[Existence of pure-strategy Nash equilibria]  If the game \eqref{value} is a continuous exact potential game with compact  $\mathcal{S}$, then the game {has}  at least one pure-strategy Nash equilibrium. Moreover, if the potential function $F^t$ is strictly convex, then the pure-strategy Nash equilibrium is unique. 
\end{lemma}
\begin{proof}
    See Appendix C. 
\end{proof}

{\subsection{Constructing continuous potential games for autonomous driving}}

In the next theorems, we show that the cost function shaping approaches developed in Section \ref{sec:finite} also work in continuous potential games. Theorem \ref{t4} considers self-centered driving objectives, Theorem \ref{t5} considers pairwise interactions, and Theorem \ref{t6} considers a mix of these two types of costs.

\begin{thm}\label{t4}
Assume that the cumulative cost in \eqref{value} are everywhere differentiable on an open superset of $\mathcal{S}$, and 
\begin{equation}
    V_i^t(\mathbf{u}_i(t),\mathbf{u}_{-i}(t))=V_i^{t,self}(\mathbf{u}_i(t)),
\end{equation}
where $V_i^{t,self}(\mathbf{u}_i(t))$ is {dependent} solely {on} $\mathbf{u}_i(t)$. Then the game \eqref{value} is a continuous potential game with the potential function,
\begin{equation}\label{p_c_1}
    P^t(\mathbf{u}(t))=\sum_{i\in\mathcal{N}}V_i^{t,self}(\mathbf{u}_i(t)).
\end{equation}
\end{thm}
\begin{proof}
The following equation holds.
\begin{equation}
\begin{split}
    \frac{\partial P^t(\mathbf{u}_i(t),\mathbf{u}_{-i}(t))}{\partial \mathbf{u}_i(t)}&=\frac{\partial \sum_{j\in\mathcal{N}}V_j^{t,self}(\mathbf{u}_j(t))}{\partial \mathbf{u}_i(t)}\\
    &=\frac{\partial V_i^{t,self}(\mathbf{u}_i(t))}{\partial \mathbf{u}_i(t)}\\
   &=\frac{\partial V_i^t(\mathbf{u}_i(t),\mathbf{u}_{-i}(t))}{\partial \mathbf{u}_i(t)}.
    \end{split}
\end{equation}
As such, according to Definition \ref{d5}, the game is a continuous  potential game, and $P^t$ is a potential function.
\end{proof}

\begin{thm}\label{t5}
Assume that the cumulative cost in \eqref{value} are everywhere  differentiable on an open superset of $\mathcal{S}$, and
\begin{equation}
    V_i^t(\mathbf{u}_i(t),\mathbf{u}_{-i}(t))=\sum_{j\in\mathcal{N},j\neq i} V_{ij}^t(\mathbf{u}_i(t),\mathbf{u}_{j}(t)),
\end{equation}
where $V_{ij}^t(\mathbf{u}_i(t),\mathbf{u}_{j}(t))=V_{ji}^t(\mathbf{u}_j(t),\mathbf{u}_{i}(t)), \forall i,j\in\mathcal{N}, i\neq j, \forall \mathbf{u}_i(t)\in \mathcal{S}_i,\forall \mathbf{u}_{j}(t)\in \mathcal{S}_j$. 
Then the game \eqref{value} is a continuous  potential game with the potential function, 
\begin{equation}\label{p_c_2}
    G^t(\mathbf{u}(t))=\sum_{i\in \mathcal{N}}\sum_{j\in\mathcal{N},j< i} V_{ij}^t(\mathbf{u}_i(t),\mathbf{u}_{j}(t)).
\end{equation}
\end{thm}
\begin{proof}
Because $V_{ij}^t(\mathbf{u}_i(t),\mathbf{u}_{j}(t))=V_{ji}^t(\mathbf{u}_j(t),\mathbf{u}_{i}(t)), \forall i,j\in\mathcal{N}, i\neq j$, the following equation holds. 
\begin{equation}
\begin{split}
   & \frac{\partial G^t(\mathbf{u}_i(t),\mathbf{u}_{-i}(t))}{\partial \mathbf{u}_i(t)}=\frac{\partial \left(\sum_{i\in \mathcal{N}}\sum_{j\in\mathcal{N},j< i} V_{ij}^t(\mathbf{u}_i(t),\mathbf{u}_{j}(t)\right)}{\partial \mathbf{u}_i(t)}\\
    &=\frac{\partial \left(\sum_{j\in\mathcal{N},j\neq i} V_{ij}^t(\mathbf{u}_i(t),\mathbf{u}_{j}(t)\right)}{\partial \mathbf{u}_i(t)}=\frac{\partial V_i^t(\mathbf{u}_i(t),\mathbf{u}_{-i}(t))}{\partial \mathbf{u}_i(t)}.
    \end{split}
\end{equation} 
As such, the game is a continuous potential game with $G^t$ as the potential function.
\end{proof}

\begin{thm}\label{t6}
Assume that the cumulative cost in \eqref{value} are everywhere differentiable on an open superset of $\mathcal{S}$, and
\begin{equation}\label{design_value_c}
\begin{split}
    &V_i^t(\mathbf{u}_i(t),\mathbf{u}_{-i}(t))\\
    &=\alpha V_i^{t,self}(\mathbf{u}_i(t))+\beta\sum_{j\in\mathcal{N},j\neq i} V_{ij}^t(\mathbf{u}_i(t),\mathbf{u}_{j}(t)),
    \end{split}
\end{equation}
where $V_i^{t,self}(\mathbf{u}_i(t))$ and $V_{ij}^t(\mathbf{u}_i(t),\mathbf{u}_{j}(t))$ satisfy the requirements in Theorems \ref{t4} and \ref{t5}, respectively. Then the game \eqref{value} is a continuous  potential game with the potential function, 
\begin{equation}\label{potentialfunction_c}
\begin{split}
     F^t(\mathbf{u}(t))&=\alpha P^t(\mathbf{u}(t))+\beta G^t(\mathbf{u}(t))\\
    & ={\alpha} \sum_{i\in\mathcal{N}}V_i^{t,self}(\mathbf{u}_i(t))\\
    &\quad+{\beta}\sum_{i\in \mathcal{N}}\sum_{j\in\mathcal{N},j< i} V_{ij}^t(\mathbf{u}_i(t),\mathbf{u}_{j}(t)).
\end{split}
\end{equation}
\end{thm}
\begin{proof}
This theorem is straightforward by combining Theorems \ref{t4} and \ref{t5}.
\end{proof}

{\subsection{Solving the constructed continuous potential games}}
As the action space contains infinitely many elements in continuous potential games, the convergence of the best response dynamics may not be achieved within {a} finite {number of iterations}. To address this issue, we define {an} $\epsilon$-Nash equilibrium.

\begin{defn} [$\epsilon$-Nash equilibrium \cite{game_book}] \label{d2}
An $N$-tuple of strategies $\{\mathbf{u}_1^*,\mathbf{u}_2^*,...,\mathbf{u}_N^*\}$ is  an $\epsilon$-Nash equilibrium for an $N$-player game if and only if $\exists \epsilon \in \mathbb{R}_{+}$ such that, $\forall i\in \mathcal{N}$, $\forall \mathbf{u}_i(t)\in\mathcal{S}_i$,   
\begin{equation}\label{eNash_de}
 V_i^t(\mathbf{u}_i^*(t),\mathbf{u}_{-i}^*(t))\leq V_i^t(\mathbf{u}_i(t),\mathbf{u}_{-i}^*(t))+\epsilon.
\end{equation}
\end{defn}

{With the $\epsilon$-Nash equilibrium, we can use the best response dynamics, i.e., Algorithm \ref{A1}, or the potential function optimization, i.e., Algorithm \ref{A2}, to solve the game, if appropriate changes are made. Specifically,  to use Algorithm \ref{A1} and to ensure the convergence within a finite number of iteration,  the $NashCondition$ in Procedure 7 should be replaced by ``$    V_j^t(\mathbf{u}^*_j(t),\mathbf{u}^*_{-j}(t))\leq V_j^t(\mathbf{u}'_j(t),\mathbf{u}^*_{-j}(t))+\epsilon$ holds $\forall \mathbf{u}'_j\in \mathcal{S}_j, \forall j\in\mathcal{N}$"; and the update rule in Procedure 5 should be replaced by ``Update $\mathbf{u}_j(t)$ using $\mathbf{u}^*_j(t)$ if $V_j^t(\mathbf{u}^*_j(t),\mathbf{u}_{-j}(t))\leq V_j^t(\mathbf{u}_j(t),\mathbf{u}_{-j}(t))-\epsilon$". In Algorithm \ref{A2}, the potential function $F^t$ should be replaced by \eqref{potentialfunction_c}.}

Next lemma shows the guaranteed convergence of the best response algorithm. 

\begin{thm}\label{l6}[Convergence to $\epsilon$-Nash equilibrium] Consider the game \eqref{value} and the Algorithm \ref{A1}. If the cumulative cost $V_i^t$ in  \eqref{value} satisfies \eqref{design_value_c} with   $\mathcal{S}$ being a connected and compact set, and if Algorithm \ref{A1} is  modified according to the $\epsilon$-Nash equilibrium condition  as specified above,  then Algorithm \ref{A1} converges within a finite number of iterations.
\end{thm}
\begin{proof}
   According to Theorem \ref{t6}, if  the cumulative cost $V_i^t$ in  \eqref{value} satisfies \eqref{design_value_c}, and if $\mathcal{S}$ is a connected set, then the game \eqref{value} is a continuous potential game. Because the potential function $F^t$ is continuous in the compact set $\mathcal{S}$, $F^t$ is bounded. Therefore, the $\epsilon$-improvement path generated by Algorithm \ref{A1} must terminate within a finite number of iterations because the $\epsilon$-increment of the potential function is finite. When an $\epsilon$-improvement path comes to an end point, there exists no strategy profile that leads to a decrease of cumulative cost greater than $\epsilon$ for any player. As such, this end point is an $\epsilon$-Nash equilibrium. 
\end{proof}

\section{Numerical Results}\label{simulation}
In this section, we apply  the developed {finite and continuous potential game frameworks} to specific traffic scenarios, including intersection-crossing scenarios and lane-changing scenarios.

{The vehicles' dynamics are described by a kinematic bicycle model \cite{bicycle_2,bicycle_validation}:
\begin{equation}\label{dynamics}
\begin{split}
    x_i(t+1)&=x_i(t)+   v_{i}(t)\text{cos}(\phi_i(t)+\beta_i(t))\Delta t,\\
    y_i(t+1)&=y_i(t)+ v_{i}(t)\text{sin}(\phi_i(t)+\beta_i(t))\Delta t,\\
    v_i(t+1)&=v_i(t)+ a_{i}(t)\Delta t,\\
    \phi_{i}(t+1)&=\phi_{i}(t)+ \frac{v_i(t)}{l_r}\text{sin}(\beta_i(t))\Delta t,\\
    \beta_i(t+1)&=\beta_i(t)+\text{tan}^{-1}\left(\frac{l_r}{l_r+l_f}\text{tan}(\delta_{i,f}(t))\right)\Delta t,
    \end{split}
\end{equation}
where $i=1,\cdots,N$ designates the $i^{th}$ agent; $x_i$ and $y_i$ are the longitudinal and lateral position of the center of mass of  vehicle $i$ along $x$ and $y$ axes, respectively; $v_{i}$ and  $a_{i}$  are  the velocity and acceleration of the center of mass of  vehicle $i$, respectively; $\beta_i$ is  the angle of the velocity with respect to the longitudinal axis of the vehicle; $\phi_i$ is the inertial heading; $l_f$ and $l_r$ are the lengths from the center of mass to the front and rear ends of the car, respectively, and are selected as $l_f = l_r = 1.5$ $m$; $\delta_{i,f}$ is the steering angle of the front wheels, and $|\delta_{i,f}|\leq 20^{\circ}$. Because for most vehicles, the
rear wheels cannot be steered, the rear wheel steering angles are assumed to be zero \cite{bicycle_2}. The inputs of each vehicle are the acceleration $a_i(t)$ and the steering angle $\delta_{i,f}(t)$.}

The following numerical studies are conducted:  Section \ref{sub_1}  applies the finite potential game framework to lane-changing scenarios. Section \ref{sub_2}  applies the continuous potential game framework  to  intersection-crossing scenarios. Section \ref{sub_3}  compares the performance of various algorithms  in intersection-crossing scenarios, {and Section \ref{sub_4} compares potential games with other decision-making approaches including reinforcement learning and control barrier function based approaches.}  

\subsection{Finite potential game in lane-changing scenarios}\label{sub_1}
Consider the three-lane highway scenario pictured in Figure \ref{lane_changing_figure}. When driving in highways, a vehicle needs to determine  1) whether to accelerate or decelerate, and 2) which lane to drive in. Therefore, we define the action space of each vehicle as
\begin{equation}\label{finite_space}  \mathcal{U}_i=\mathcal{A}_i\times\mathcal{L}_i,
\end{equation}
where 
\begin{equation}\nonumber
\begin{split}
    \mathcal{A}_i=&\{\text{hard brake, brake, mild brake, maintain speed, }\\ &\quad\text{mild accelerate, accelerate, rapid accelerate}\},
\end{split}
\end{equation}
\begin{equation}\nonumber
\begin{split}
\mathcal{L}_i=&\{\text{change to the left lane, }\text{maintain the current lane, }\\
&\quad\text{change to the right lane}\}.
\end{split}
\end{equation}
{In \eqref{finite_space}, both $\mathcal{A}_i$ and $\mathcal{L}_i$ are finite sets, and therefore, $\mathcal{U}_i$ is a finite action space.} We let $a_i(t)=-3$ $m/s^2$ (resp. $3$ $m/s^2$) if ``{hard brake}" (resp. ``{rapid accelerate}") is selected, $a_i(t)=-2$ $m/s^2$ (resp. $2$ $m/s^2$) if ``{brake}" (resp. ``{accelerate}") is selected, $a_i(t)=-1$ $m/s^2$ (resp. $1$ $m/s^2$) if ``{mild brake}" (resp. ``{mild accelerate}") is selected, and $a_i(t)=0$ $m/s^2$ if ``{maintain speed}" is selected. If ``change to the left lane" (resp. ``{change to the right lane}") is selected, then a constant steering angle  $\delta_{i,f}(t)=0.9^{\circ}$ (resp. $-0.9^{\circ}$) is added to the steering angle required to negotiate the road, if any, until the AV is in the center of the target lane. Once the AV reaches the center line, a lane centering controller is activated at $10$ $Hz$ sampling rate, applying the steering required to reduce the relative heading between the vehicle and the road. Specifically, the steering angle induces the vehicle side slip angle $\beta_i(t)$ in the amount of $-\frac{1}{2}(\phi_i(t)-\phi_{r}(t))$, where $\phi_{r}(t)$ depends on the road orientation. On the other hand, if the decision is ``stay in the current lane", then $\delta_{i,f}(t)$ remains  at the steering angle corresponding to the curvature of the road (or that of the turning lanes in an intersection). 

In our simulation, we use a simplified strategy space: $\mathbf{u}_i(t)=\{u_i(t),u_i(t+1),...,u_i(t+T-1)\}$, where $u_i(t+\tau)=u_i(t)\in\mathcal{U}_i$, $\forall \tau=0,1,...,T-1$. This simplified strategy space represents that the AV plans to continue the same action over the horizon $T$. This assumption is reasonable as it is consistent with a common driving experience. When a human driver plans a maneuver (e.g., lane-changing),  he/she usually expects to continue the maneuver for some time to complete it, instead of planning a sequence of different maneuvers for every $\Delta t$. Note that although the AV {plans} one maneuver for the next $T$ periods, it may change its mind and select  a different maneuver after $\Delta t$, triggered by the change of its environment {and} by the receding horizon optimization  conducted at every $t$. {In the simulation, we select $T$ to correspond to $4$ $s$  to cover a lane-changing maneuver duration \cite{add_T} and select $\Delta t=0.5$ $s$ to enable the ego vehicle to respond timely to the change of  its surrounding vehicles' states and behaviors \cite{add_deltaT}. Note that a larger $T$ and a smaller $\Delta t$ lead to increased computational time and effort for the decision-making algorithm.}

In such a multi-vehicle scenario, the  driving performance of the ego {vehicle} (labeled with the number "1" in Figure \ref{lane_changing_figure})   is correlated with its surrounding vehicles' states (i.e., positions and velocities) and actions (i.e., lane-changing and/or accelerations). To take into {consideration} such interactions, at every $t$, the ego {vehicle} solves {an} $N$-player game as described in \eqref{value}  to generate its optimal strategy.  We set $N=5$  in our simulations and let vehicles' initial positions be similar to the ones in Figure \ref{lane_changing_figure}. 

\begin{figure}[thpb]
\centering
\includegraphics[width=0.35\textwidth]{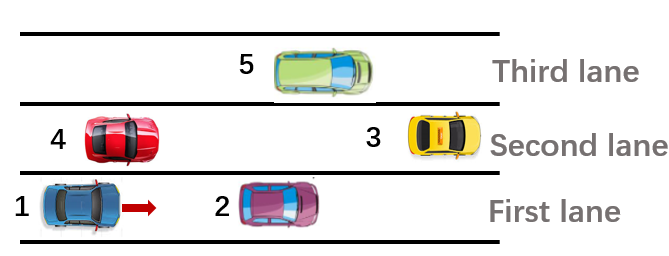}
\caption{The lane-changing scenario. }\label{lane_changing_figure}
\end{figure}

The  cumulative cost of each vehicle is {designed} as
\begin{equation}\label{cost_1}
\begin{split}
    &V_i^t(\mathbf{u}_i(t),\mathbf{u}_{-i}(t))\\
    &=\theta_{i,1}V_i^{t,self,1}(\mathbf{u}_i(t))+\theta_{i,2}V_i^{t,self,2}(\mathbf{u}_i(t))\\
    &\quad+\theta_{i,3}\sum_{j\in \mathcal{N},j\neq i}V^t_{ij}(\mathbf{u}_i(t),\mathbf{u}_j(t)),
\end{split}
\end{equation}
where $\theta_{i,1}$, $\theta_{i,2}$ and $\theta_{i,3}$ are constant parameters to tune vehicle $i$'s driving styles or aggressiveness. {In our simulations, we select $\theta_{i,1}=\theta_{i,2}=1$ and $\theta_{i,3}=4$ for all $i\in\mathcal{N}$. Equation \eqref{cost_1} facilitates a framework for modeling the surrounding vehicles with a variety of human driver behaviors, where the parameters can be learned or calibrated from naturalistic traffic dataset \cite{payoff1,add_collect} using inverse optimal control or inverse reinforcement learning  \cite{inverse_RL}.} The first term $V_i^{t,self,1}(\mathbf{u}_i(t))$, is to motivate the vehicle to track the desired speed, and is given by
\begin{equation}\label{trackspeed}
     V_i^{t,self,1}(\mathbf{u}_i(t))=\sum_{\tau=t}^{t+T-1} \left(\frac{v_i(\tau)-v_{i,d}}{v_{i,d}}\right)^2,
\end{equation}
where $v_{i,d}$ is the desired speed of vehicle $i$. {For the ego vehicle, we let $v_{1,d}=27$ $m/s$ (approximately  $60$ $mph$). For the surrounding vehicles, we let $v_{i,d}$ be within the range $v_{i,d}\in[22,30]$ $m/s$ for $i=2,3,4,5$. (Multiple  scenarios  with various  ego  
vehicle  desired  speed  in  the  range  of  $[10,30]$ $m/s$ are tested, and the ego vehicle performance remains similar in all tested scenarios.)} The second term $V_i^{t,self,2}(\mathbf{u}_i(t))$ is to prevent the vehicle from driving outside the road, and is given by
\begin{align}
&V_i^{t,self,2}(\mathbf{u}_i(t))=\sum_{\tau=t}^{t+T-1}\gamma\cdot\mathbf{1}_{\{y_i(\tau)<0 \text{ or } y_i(\tau)>W_d\}}
\end{align}
where $W_d$ is the width of the road, $\mathbf{1}_{\{condition\}}=1$ if $condition$ is true and $0$ otherwise.  {The constant $\gamma$ represents the penalty for driving out of the road and should be a sufficiently large number. We select $\gamma=1000$ in our simulation while noting that the AV behavior is relatively insensitive to this value as long as it is chosen sufficiently large.} The third term $\sum_{j\in \mathcal{N},j\neq i}V_{ij}^t(\mathbf{u}_i(t),\mathbf{u}_j(t))$ is to avoid  collision:
\begin{equation}\label{collision_design}
    V_{ij}^t(\mathbf{u}_i(t),\mathbf{u}_j(t))=\sum_{\tau=t}^{t+T-1}J_{ij}(x_i(\tau),y_i(\tau),x_j(\tau),y_j(\tau)),
\end{equation}
and
{\begin{equation}\label{tanh_cost}
\begin{split}
    &J_{ij}(x_i(\tau),y_i(\tau),x_j(\tau),y_j(\tau))\\
    &=\left(\tanh{\left(\beta(d^2_{x,c}-\left(x_i(t)-x_j(t)\right)^2)\right)}+1\right)\\
    &\quad\cdot\left(\tanh{(\beta(d^2_{y,c}-\left(y_i(t)-y_j(t)\right)^2))}+1\right),
    \end{split}
\end{equation}
where $d_{x,c}$ and $d_{y,c}$ are the longitudinal and lateral collision distances, respectively, and $d_{x,c}=7$ $m$ and $d_{y,c}=4.5$ $m$ in the simulation. The parameter $\beta$ should be sufficiently large such that the $tanh$ function always takes the two extreme values $-1$ or $1$, and we set $\beta=1000$ in the simulation. The cost \eqref{tanh_cost} means that at time $\tau$, if vehicle $j$ is a threat to vehicle $i$ both laterally and longitudinally, then $J_{ij}(x_i(\tau),y_i(\tau),x_j(\tau),y_j(\tau))\approx4$; Otherwise, $J_{ij}(x_i(\tau),y_i(\tau),x_j(\tau),y_j(\tau))\approx0$. Note that the lane width is $5$ $m$, and {the lateral collision distance $d_{y,c}$ is selected to be a bit smaller than the road width to realize the feature that if a vehicle is not in the same lane or in the target lane  of the ego vehicle (target lane refers to the lane that the ego vehicle aims to change to), then it should not be considered as a threat to the ego vehicle. The necessity of this feature is illustrated by a comparative study at \textit{https://youtu.be/aWXTfLdovfc}.}   
}

According to Theorem \ref{t3}, the above cost function design makes the $N$-player game a finite potential game. Therefore, we can employ the potential function optimization algorithm, i.e., Algorithm \ref{A2}, to solve the game. We let the ego {vehicle} {adopt} the generated NE as its strategy, {assuming that the surrounding vehicles also behave to optimize their own driving performances characterized by \eqref{cost_1}. As the actual behavior of the surrounding vehicles may  deviate from the assumed one, we perform statistical studies incorporating various surrounding vehicle behaviors in Section \ref{sub_3}. In the specific scenarios 1-4 listed below,} the surrounding vehicles use  non-intelligent strategies: maintaining a constant speed and heading. These simulation studies illustrate empirically the robustness of potential game approach against various surrounding vehicles' strategies. 

 \textbf{Scenario 1}. In this scenario, we let  vehicle $2$ (resp. vehicle $3$) move slower (resp. faster) than the ego {vehicle}. As the front vehicle moves slowly, the ego {vehicle} decides to perform lane-changing when it is safe. After changing to the second lane, the ego {vehicle} then keeps its desired speed and drives in the second lane for the rest of the simulation time. {The ego vehicle trajectory is shown in Figure \ref{lanechanging_1_1}, its velocity and  control inputs are shown in Figures \ref{lanechanging_1v} and \ref{lanechanging_1u}, respectively, and the full animation is provided as ``Scenario 1" at \textit{https://youtu.be/nQkpdQRcwEE}.}

\begin{figure}[thpb]
\centering
\subfigure[]{\label{lanechanging_1_1}
\includegraphics[width=0.36\textwidth]{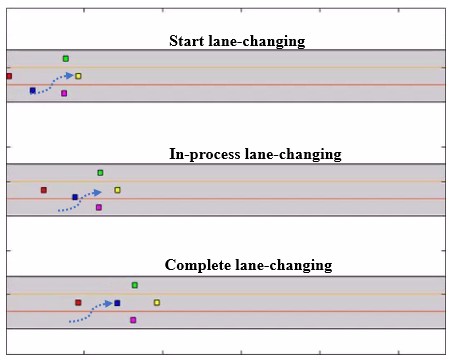}}
\subfigure[]{\label{lanechanging_1v}
\includegraphics[width=0.23\textwidth]{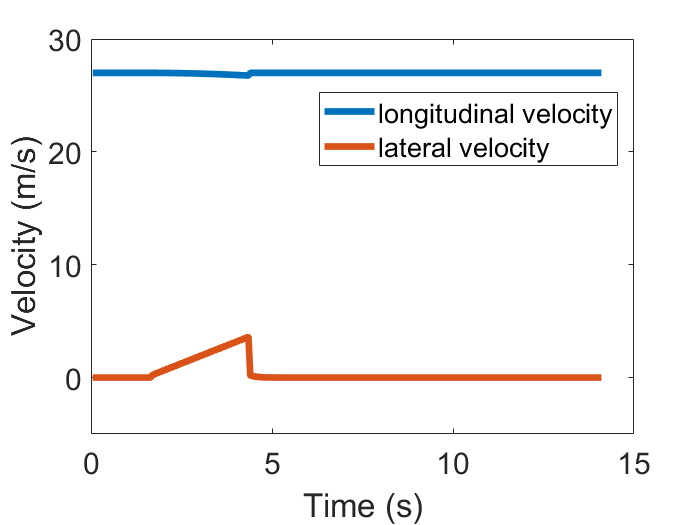}}
\subfigure[]{\label{lanechanging_1u}
\includegraphics[width=0.23\textwidth]{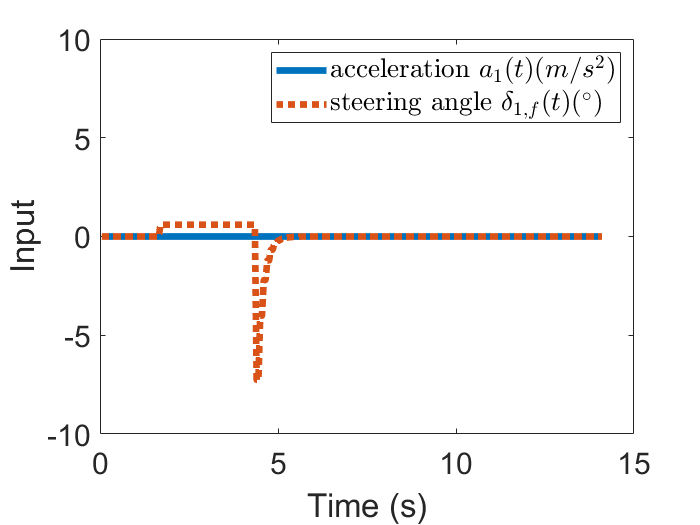}}
{\caption{The ego {vehicle}  behavior in Scenario 1. (a) The ego vehicle trajectory. {(b) The ego vehicle velocity. (c) The ego vehicle acceleration and steering angle.}}}\label{lanechanging_1}
\end{figure}

\textbf{Scenario 2.} In this scenario, we let both vehicles $2$ and $3$ move slower than the ego {vehicle}, and vehicle $5$ moves faster. In this case, the ego {vehicle} is no longer  satisfied driving in the second lane, and as a result, it performs two consecutive lane-changing when it is safe. {The ego vehicle behavior is shown in Figure \ref{lanechanging_2}, and the full animation is provided as ``Scenario 2" at \textit{https://youtu.be/nQkpdQRcwEE}.}
\begin{figure}[thpb]
\centering
\subfigure[]{\label{lanechanging_2_1}
\includegraphics[width=0.36\textwidth]{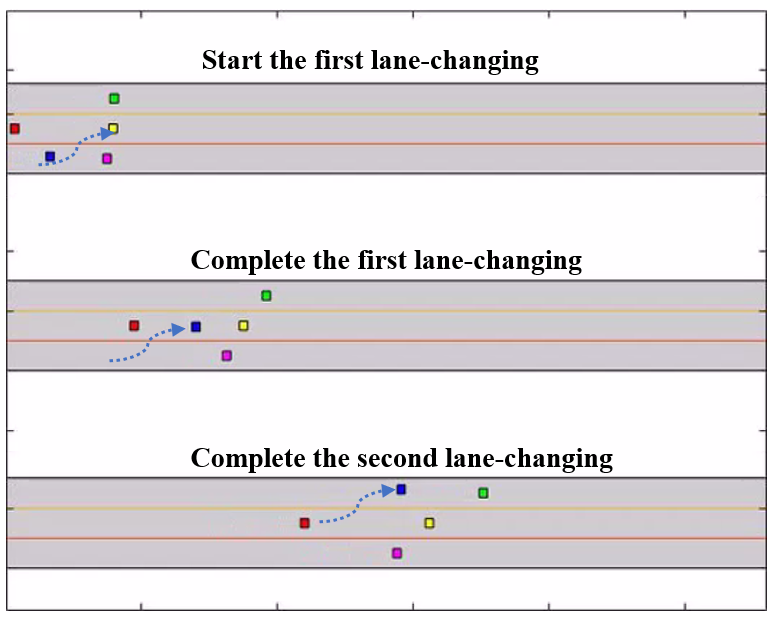}}
\subfigure[]{\label{lanechanging_2v}
\includegraphics[width=0.23\textwidth]{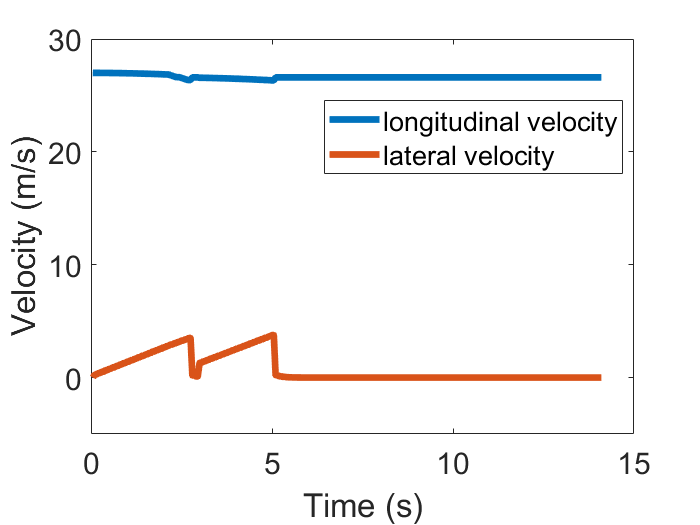}}
\subfigure[]{\label{lanechanging_2u}
\includegraphics[width=0.23\textwidth]{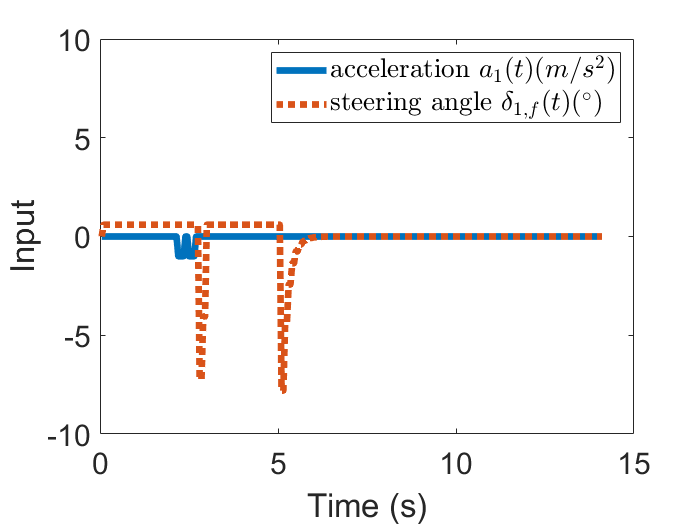}}
{\caption{The ego  {vehicle}  behavior in Scenario 2. (a) The ego  {vehicle} trajectory. {(b) The ego vehicle velocity. (c) The ego vehicle acceleration and steering angle.}}}\label{lanechanging_2}
\end{figure}

\textbf{Scenario 3.} In this scenario, we let all three  vehicles in-front, i.e., vehicles $2$, $3$, and $5$, move slower than the ego  {vehicle} desired speed. As {these} three vehicles block all of the three lanes, it is hard for the ego  {vehicle} to keep its desired speed. {Based on the game theoretic strategy, the ego  {vehicle} chooses actions that result} in both desired speed and safety. It first changes to the second lane and follows vehicle $3$. {After it overtakes vehicle $2$, the ego  {vehicle} then changes back to the first lane to maintain its desired speed.} In this scenario,  the ego  {vehicle} intelligently prepares itself for the opportunity to overtake vehicle $2$, indicating that the ego  {vehicle} has the ability to plan {actions over} the long run. {The ego vehicle trajectory and behavior are shown in Figure \ref{lanechanging_3}, and the full animation is shown in ``Scenario 3" at \textit{https://youtu.be/nQkpdQRcwEE}.}
\begin{figure}[thpb]
\centering
\subfigure[]{\label{lanechanging_3_1}
\includegraphics[width=0.36\textwidth]{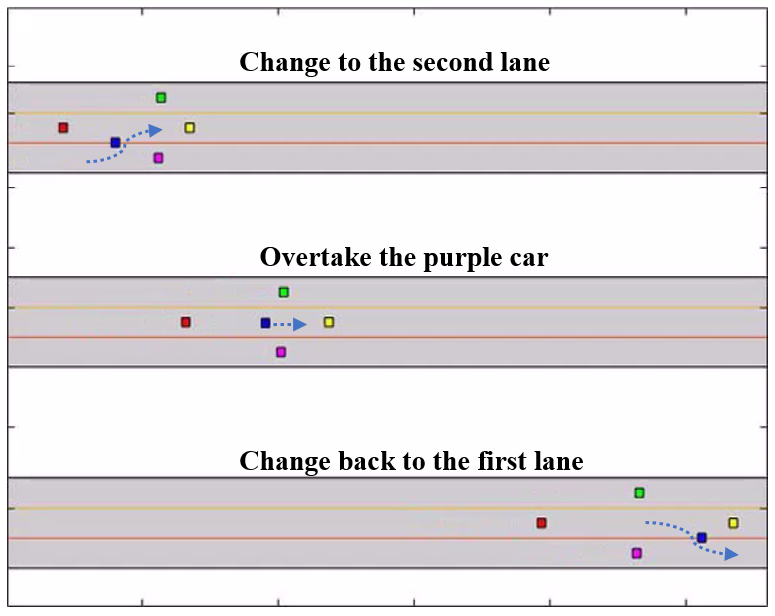}}
\subfigure[]{\label{lanechanging_3v}
\includegraphics[width=0.23\textwidth]{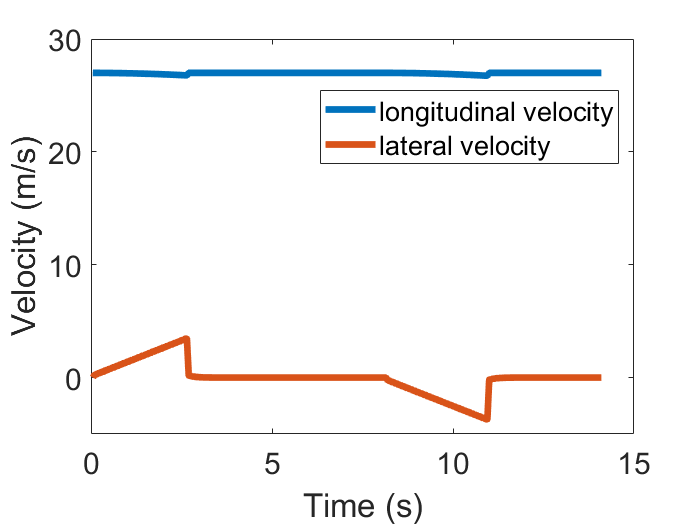}}
\subfigure[]{\label{lanechanging_3u}
\includegraphics[width=0.23\textwidth]{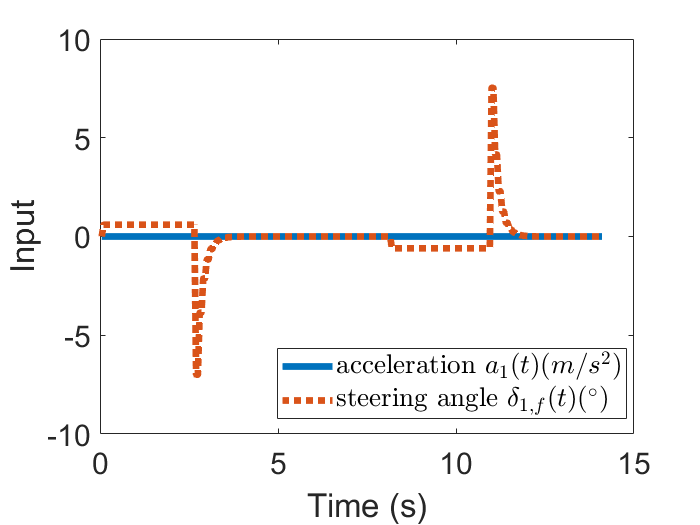}}
{\caption{The ego {vehicle}  behavior in Scenario 3. (a) The ego {vehicle} trajectory. {(b) The ego vehicle velocity. (c) The ego vehicle acceleration and steering angle.}}}\label{lanechanging_3}
\end{figure}

{\textbf{Scenario 4.} In this scenario, we test the ego vehicle lane-changing behavior on curved roads. The vehicles' initial positions and velocities are similar to Scenario 1, i.e., vehicle $2$ moves slower than the ego vehicle, and vehicle $3$ moves faster. To keep the desired speed, the ego vehicle changes to the second lane when it is safe, as shown in Figure \ref{lanechanging_4}. The full animation is available as ``Scenario 4" at \textit{https://youtu.be/nQkpdQRcwEE}.}

\begin{figure}[thpb]
\centering
\subfigure[]{\label{lanechanging_4_1}
\includegraphics[width=0.32\textwidth]{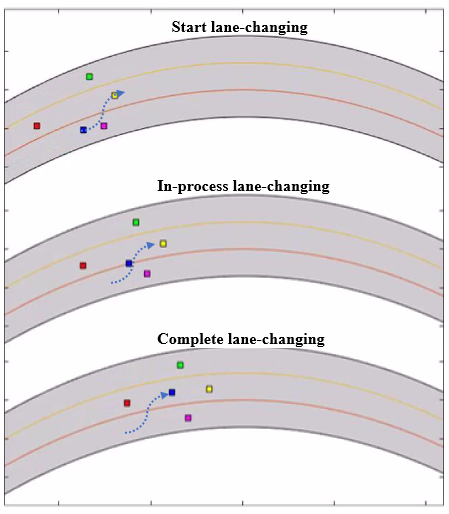}}
\subfigure[]{\label{lanechanging_4v}
\includegraphics[width=0.23\textwidth]{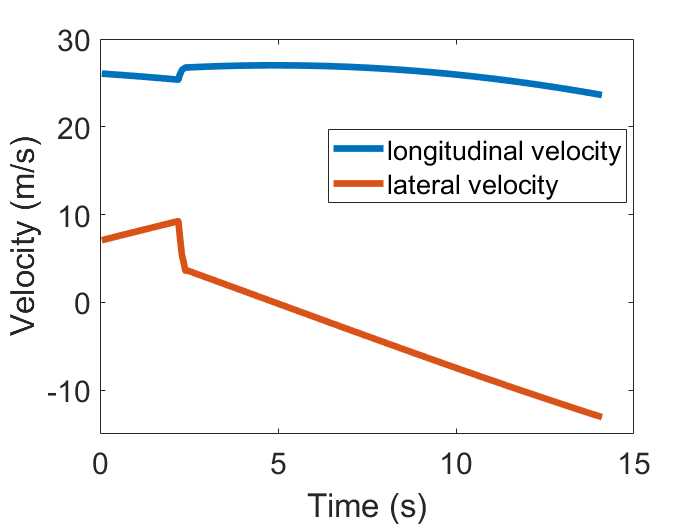}}
\subfigure[]{\label{lanechanging_4u}
\includegraphics[width=0.23\textwidth]{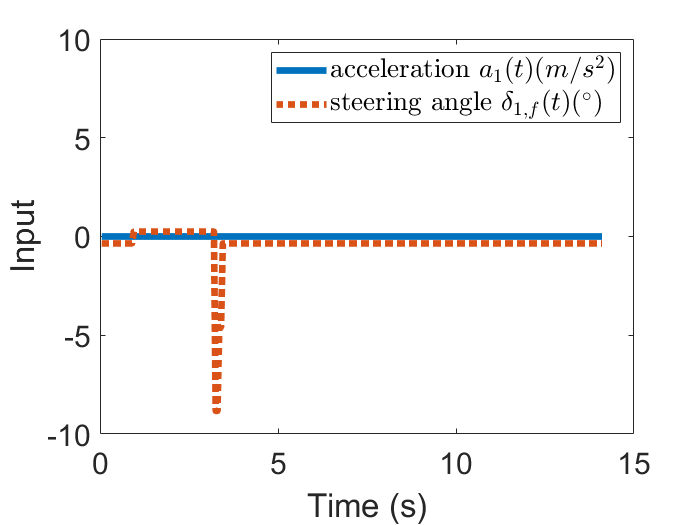}}
{\caption{The ego {vehicle}  behavior in Scenario 4. (a) The ego vehicle trajectory. {(b) The ego vehicle velocity. (c) The ego vehicle acceleration and steering angle.}}}\label{lanechanging_4}
\end{figure}

\subsection{Continuous potential game in intersection-crossing scenarios}\label{sub_2}
In this section, we apply the continuous potential game framework developed in Section \ref{sec_continuous} {to} intersection-crossing scenarios. 

Consider the  scenario pictured in Figure \ref{intersection}. The vehicle labeled with the number "1" is the ego  {vehicle}, aiming to go straight and cross the intersection. The moving directions of its surrounding vehicles (i.e., vehicles $2$-$5$) are assumed to be known to the ego vehicle, and they are marked using grey arrows in Figure \ref{intersection}. {We let the action space of each vehicle be {$\mathcal{U}_i=[-3,3]$ $m/s^2$}, representing that the vehicle acceleration $a_i(t)$ can be any real number from the interval $[-3,3]$.} We also employ the simplified strategy space as illustrated in Section \ref{sub_1}.

The cumulative cost {function} of each vehicle is {chosen} as 

\begin{equation}\label{cost_2}
\begin{split}
   & V_i^t(\mathbf{u}_i(t),\mathbf{u}_{-i}(t))\\&=\theta_{i,1}V_i^{t,self,1}(\mathbf{u}_i(t))
    +\theta_{i,2}\sum_{j\in \mathcal{N},j\neq i}V_{ij}^t(\mathbf{u}_i(t),\mathbf{u}_j(t))+\varpi_i,
\end{split}
\end{equation}
where $\varpi_i$ is a small positive  noise to prevent vehicles from getting stuck in case of identical cost and initial conditions, $V_i^{t,self,1}(\mathbf{u}_i(t))$ takes the form \eqref{trackspeed} to motivate vehicles to track their desired speeds, and $V_{ij}^t,(\mathbf{u}_i(t),\mathbf{u}_j(t))$ takes the form \eqref{collision_design} with $J_{ij}$ designed as
\begin{equation}\label{cost_inverse}
\begin{split}
     &J_{ij}(x_i(\tau),y_i(\tau),x_j(\tau),y_i(\tau))\\
     &\quad= \frac{1}{\left(x_i(\tau)-x_j(\tau)\right)^2+\left(y_i(\tau)-y_j(\tau)\right)^2+\delta},
\end{split}   
\end{equation}
if vehicles $i$ and $j$ are in conflict (i.e., their planned trajectories intersect in the intersection region), and $J_{ij}(x_i(\tau),y_i(\tau),x_j(\tau),y_i(\tau))=0$ if $i$ and $j$ are not in conflict. {This inverse-type cost function enables the ego vehicle to adjust its speed smoothly and continuously in the intersection-crossing scenarios. The parameter $\delta$ in \eqref{cost_inverse} is used to avoid the denominator being zero during the simulation, and  $\delta=0.01$.}

According to Theorem \ref{t6}, the $N$-player game is a continuous potential game with the above  cost function design. 


\begin{figure}[thpb]
\centering
\includegraphics[width=0.3\textwidth]{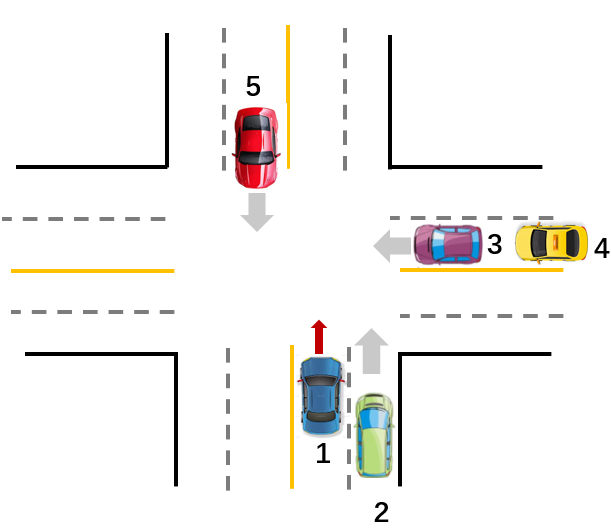}
\caption{The lane-changing scenario }\label{intersection}
\end{figure}

 Figure \ref{intersection_1} shows the performance of the continuous potential game in a specific scenario, where all vehicles follow the NE strategies. The speed of each vehicle is shown {in} the top left corner in Figures \ref{intersection_1_1} and \ref{intersection_1_2}. The ego vehicle first slows down to yield to the purple and yellow vehicles (Figure \ref{intersection_1_1}) and then speeds up after these two vehicles cross the intersection (Figure \ref{intersection_1_2}). 
 {The performance of the best response dynamics and of the potential function optimization algorithm are very close in this scenario. (Comparative animations are available at \textit{https://youtu.be/nQkpdQRcwEE}). The iterations of the best response dynamics are shown in Figures \ref{intersection_1_dy1} and \ref{intersection_1_dy2}, at the time instance illustrated in Figure \ref{intersection_1_1} and that of Figure \ref{intersection_1_2}, respectively, validating the convergence of the best response dynamics.

We also consider the scenarios where the surrounding vehicles do not take the NE strategies. In these scenarios, the potential function optimization algorithm outperforms the best response dynamics in terms of better safety performance. The specific scenarios are shown in the first part of the intersection-crossing animation at \textit{https://youtu.be/nQkpdQRcwEE}, and the statistical comparison is shown in the next subsection.}
\begin{figure}[thpb]
\centering
\subfigure[]{\label{intersection_1_1}
\includegraphics[width=0.21\textwidth]{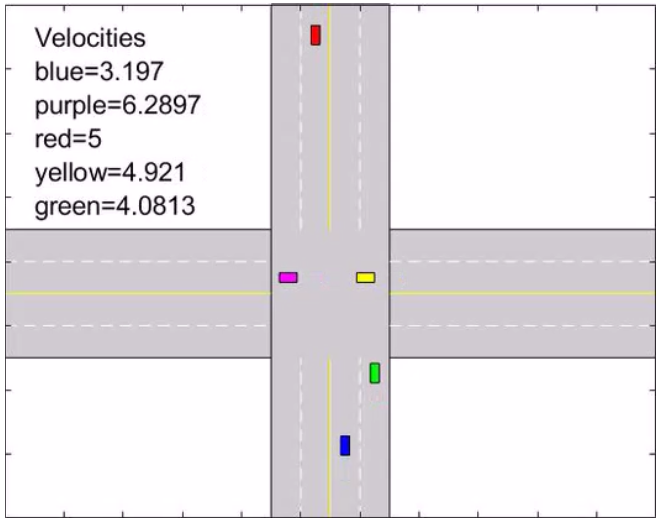}}
\subfigure[]{\label{intersection_1_2}
\includegraphics[width=0.21\textwidth]{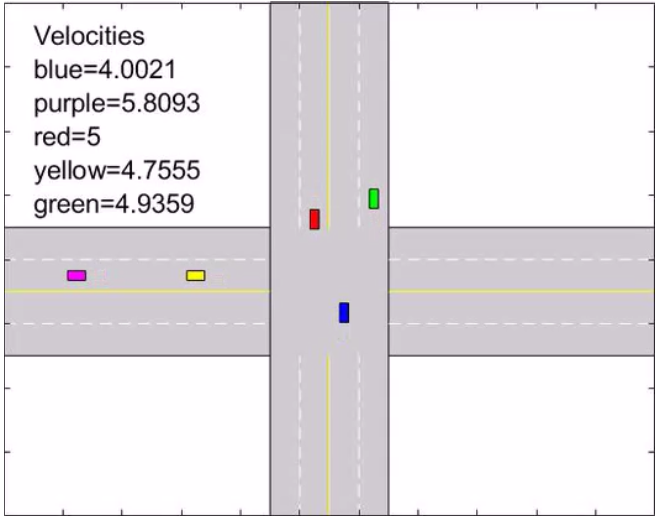}}
\subfigure[]{\label{intersection_1_dy1}
\includegraphics[width=0.234\textwidth]{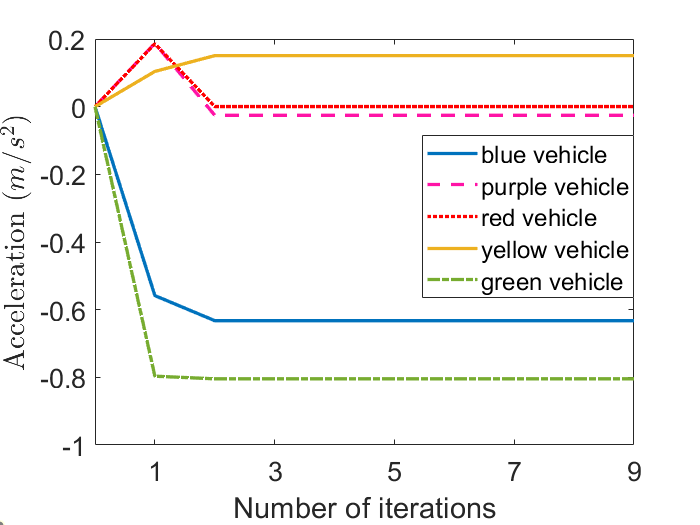}}
\subfigure[]{\label{intersection_1_dy2}
\includegraphics[width=0.234\textwidth]{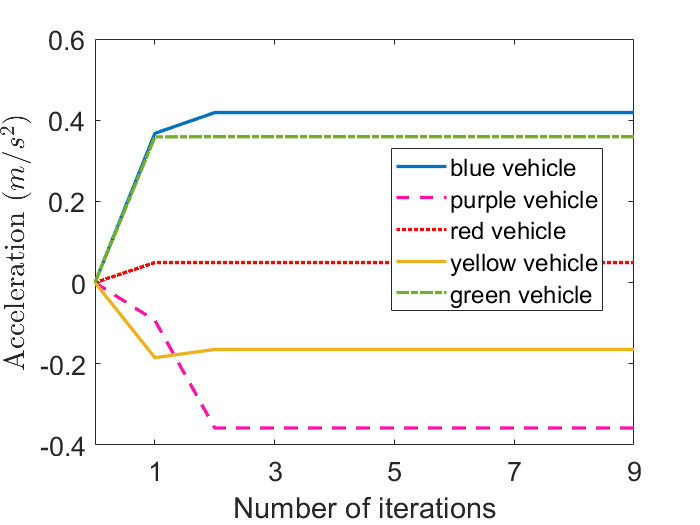}}
\caption{The ego {vehicle}  behavior in the intersection-crossing scenario, where all vehicles follow the NE strategies. (a) The ego {vehicle} slows down to avoid crashing into vehicles $3$ and $4$, which are speeding up to avoid collisions with the ego  {vehicle}. b) When it is  safe, the ego  {vehicle} speeds up to maintain its desired speed. {(c) The best response iteration at the time instance of (a). (d) The best response iteration at the time instance of (b).}}\label{intersection_1}
\end{figure}

\subsection{Statistical comparison of various potential game algorithms}\label{sub_3}
In this subsection, we conduct two comparative studies in  intersection-crossing scenarios: 1)  finite potential game vs. continuous potential game, and  2) best response dynamics vs. potential function optimization algorithm.  For each study, the comparisons are conducted in terms of a) computational efficiency measured by the average {and maximum} computational time  to complete one decision-making process, b) the ego  {vehicle}  safety measured by the  collision rate, average and maximum collision speeds, and c) the ego  {vehicle}  travel efficiency measured by the ego  {vehicle}   average speed during the  intersection-crossing. As for the surrounding vehicles,  we assign them three different strategies: a) employing the NE, b) keeping a constant speed, and c) randomly selecting an acceleration from the action space $\mathcal{U}_i$ at each $t$. The first strategy is rational and intelligent. The second strategy is non-intelligent and safety-agnostic, but may be encountered  in real scenarios as human driver may be distracted from driving and therefore cannot respond timely to a collision threat. The third strategy is neither intelligent nor rational  {but useful for testing robustness of the algorithms}.   For each algorithm, we test $5000$ intersection situations with randomly selected vehicles' initial positions  to provide reliable statistical results. In each situation,  $5$ vehicles get involved (as shown in Figure \ref{intersection}), i.e., the ego  {vehicle} solves a $5$-player game for each decision-making.  
 
 Table \ref{table1} shows the comparative results between finite potential game and continuous potential game. The running time is collected from MATLAB$^{\circledR}$ on a laptop with an Intel Core i7-10750H processor clocked at $2.60$ GHz and $16$ GB of RAM. We use the genetic algorithm function '\textit{ga}' in Matlab \cite{ga,add_ga}  to efficiently solve the non-convex optimization problems. For potential function optimization (Algorithm \ref{A2}),  $10$ start points are selected, and for best response dynamics (Algorithm \ref{A1}), $3$ start points are selected for each agent's local optimization. The number of start points is selected based on the criterion that increasing this number does not necessarily improve the ego vehicle performance (in terms of collision rate and average velocity) while decreasing this number clearly leads to worse performance.
 The collision rate represents the number of situations where a collision  happens to the ego  {vehicle} out of the total number of situations (i.e., $5000$). Concerning the average speed, note that the ego  {vehicle}  desired speed is $5m/s$, and in general, the closer the average speed  to $5$ $m/s$, the better  {the} travel efficiency  {of} the ego  {vehicle}. The average and maximum ego vehicle speed at collision indicates the contribution of the ego vehicle to the collision impact, which relates to the two vehicles' relative speed.   The action space of each vehicle is: {$\mathcal{U}_i=[-3,3]$ $m/s^2$} for continuous game and {$\mathcal{U}_i=\{-3, -2, -1, 0, 1, 2, 3\}$ $m/s^2$} for finite game.  For both potential games, the NE  is {determined by} the potential function optimization algorithm, i.e., Algorithm \ref{A1}. 
 
 \begin{table*}[!h]
\centering
\caption{Comparative results: Finite potential game vs. Continuous potential game}
\begin{tabular}{c|c|c|c|c|c|c}
\hline  
\multicolumn{1}{c|}{\multirow{2}{*}{}}
&\multicolumn{3}{c|}{\multirow{2}{*}{Finite potential game}}
&\multicolumn{3}{c}{\multirow{2}{*}{Continuous potential game}}\\
\multicolumn{1}{c|}{\multirow{2}{*}{}}
&\multicolumn{3}{c|}{\multirow{2}{*}{}}
&\multicolumn{3}{c}{\multirow{2}{*}{}}
\\
\hline
Surrounding vehicles' strategies & NE  & Constant speed & Random acceleration & NE& Constant speed & Random acceleration 
\\
\hline
\multirow{1}{*}{Collision rate} & $0/5000$ & $0/5000$ & $21/5000$ & $0/5000$ & $0/5000$ & $14/5000$
\\
\hline
\multirow{1}{*}{Average ego speed (m/s)} & $3.78$ & $3.05$ & $3.01$ & $3.93$ & $3.21$ & $3.09$
\\
\hline
\multirow{1}{*}{Avg/Max relative collision speed (m/s)} & N/A & N/A & $10.16/13.56$ & N/A & N/A & $10.54/14.68$
\\
\hline
\multirow{1}{*}{Avg/Max ego speed at collision (m/s)} & N/A & N/A & $8.06/9.50$ & N/A & N/A & $8.06/9.61$
\\
\hline
\multicolumn{1}{c|}{\multirow{1}{*}{Avg/Max computational time (s)}} &
\multicolumn{3}{c|}{\multirow{1}{*}{$0.05/ 0.10$}} &
\multicolumn{3}{c}{\multirow{1}{*}{$0.07/ 0.28$}}
\\
\hline
\end{tabular}\label{table1}
\end{table*}

Table \ref{table1} leads to the following observations:
\begin{enumerate}
    \item  Both potential game frameworks are effective in ensuring the ego  {vehicle}  safety: No collision happens in any of the $5000$ scenarios if the surrounding vehicles {follow} the NE strategies or  maintain constant speeds. Although absolute safety is not guaranteed if the surrounding vehicles randomly  accelerates at each $t$ (which is too irrational to happen in real life), the collision rate is still low (less than $1\%$). 
    \item In both games, the ego  {vehicle shows}  better travel efficiency if the surrounding vehicles {follow} the NE strategies, compared to the one of constant speeds. It implies that the more intelligent the surrounding vehicles are, the better {the} ego vehicle travel efficiency is. Meanwhile, even if the surrounding vehicles are non-intelligent (maintaining constant speeds), the ego  {vehicle}  travel efficiency remains {satisfactory} ($3.21$ $m/s$ in continuous game and $3.05$ $m/s$ in finite game), while ensuring safety.
    \item   Computational cost is affordable in both potential game frameworks (less than $0.1$ $s$ {on average}). 
    \item In general, the  continuous potential game outperforms the finite potential game, as it leads to better travel efficiency with similar safety and computation performance.
\end{enumerate}

 \begin{table*}[!h]
\centering
\caption{Comparative results: Best response dynamics vs. potential function optimization}
\begin{tabular}{c|c|c|c|c|c|c}
\hline  
\multicolumn{1}{c|}{\multirow{2}{*}{}}
&\multicolumn{3}{c|}{\multirow{2}{*}{Best response dynamics}}
&\multicolumn{3}{c}{\multirow{2}{*}{Potential function optimization}}\\
\multicolumn{1}{c|}{\multirow{2}{*}{}}
&\multicolumn{3}{c|}{\multirow{2}{*}{}}
&\multicolumn{3}{c}{\multirow{2}{*}{}}
\\
\hline
Surrounding vehicles' strategies & NE & Constant speed & Random acceleration & NE& Constant speed & Random acceleration 
\\
\hline
\multirow{1}{*}{Collision rate} & $0/5000$ & $61/5000$ & $106/5000$ & $0/5000$ & $0/5000$ & $14/5000$
\\
\hline
\multirow{1}{*}{Average ego speed (m/s)} & $4.18$ & $3.70$ & $3.58$ & $3.93$ & $3.21$ & $3.09$
\\
\hline
\multirow{1}{*}{Avg/Max relative collision speed (m/s)} & N/A & $8.84/9.62$ & $9.62/12.88$ & N/A & N/A & $10.54/14.68$
\\
\hline
\multirow{1}{*}{Avg/Max ego speed at collision (m/s)} & N/A & $7.19/8.36$ & $6.67/8.70$ & N/A & N/A & $8.06/9.61$
\\
\hline
\multicolumn{1}{c|}{\multirow{1}{*}{Avg/Max computational time (s)}} &
\multicolumn{3}{c|}{\multirow{1}{*}{$0.24/0.93$}} &
\multicolumn{3}{c}{\multirow{1}{*}{$0.07/0.28$}}
\\
\hline
\end{tabular}\label{table2}
\end{table*}

Table \ref{table2} compares the results from the two solution algorithms (best response dynamics and potential function optimization) in the continuous potential game framework. The simulation settings are selected to be the same as the ones for Table \ref{table1}.

Table \ref{table2} leads to the following observations:
\begin{enumerate}
    \item Both solution algorithms are effective in ensuring the ego  {vehicle}  safety when the surrounding vehicles {follow} the NE strategies: No collision happens in any of the $5000$ scenarios. If the surrounding vehicles are non-intelligent (i.e., maintaining constant speeds or using random accelerations), the potential function optimization algorithm performs notably better than the best response dynamics. This is because the potential function optimization finds the global optimal NE, while the best response dynamics may converge to any NE.  
    \item With regard to the ego  {vehicle}  travel efficiency, although the best response dynamics lead to higher average speeds, we cannot conclude that it has better travel efficiency. This is because the two solution algorithms have different safety levels, and although the best response dynamics result in higher speeds, it {puts} the ego  {vehicle} {into more} dangerous situations. 
    \item Computational cost is affordable in both solution algorithms ($0.07$ $s$ in potential function optimization and $0.24$ $s$ in best response dynamics, {on average}). Compared to the best response dynamics, the potential function optimization  has better computational efficiency. 
    \item In general, the   potential function optimization algorithm outperforms the best response dynamics, as it {requires} less computational time and leads to better safety  for the ego  {vehicle}.
\end{enumerate}
{
\subsection{Statistical comparison of various AV decision-making algorithms}\label{sub_4}
In this subsection, we compare potential games with other existing decision-making approaches, including reinforcement learning and control barrier function. The simulation scenarios and settings are the same as the ones in Section \ref{sub_3}, and the agents' action spaces are finite. 

For potential games, we employ the potential function optimization algorithm.

For reinforcement learning, we employ a double deep Q-network (DDQN) algorithm \cite{add_DQN_2,add_DQN_1}. Similar to the potential game, the ego vehicle's cost function in RL is also composed of two parts: One is to track the desired speed as in \eqref{trackspeed}, and the other is to avoid collision (specifically, if a collision happens, then a large penalty is applied). During the training stage, the surrounding vehicle behaviors are governed by a simplified intelligent driver model (IDM) \cite{IDM}: If vehicle $i$ senses a collision threat with vehicle $j$ (i.e., inter-vehicle distance is less than a pre-defined safe distance), then vehicle $i$ brakes if $T^c_{ij}\geq T^c_{ji}$ and accelerates otherwise, with constant a acceleration/deceleration. Here $T^c_{ij}$ represents time-to-collision for vehicle $i$ with vehicle $j$ (see \cite{add_DQN_1,mine_1} for the mathematical expression). We train the RL agent for $100000$ episodes, where each episode contains 24 samples (corresponds to the $12$ $s$ duration of the intersection crossing with $0.5$ $s$ sampling time) or less (if collisions occur). For the execution stage, we let the surrounding vehicles follow one of two strategies, NE or constant speed, to test the robustness of the performance. In this intersection-crossing scenario, the vehicle behaviors from the NE and from the IDM are very close, as they both motivate the vehicles to cooperatively avoid collisions by accelerating (resp. decelerating) the vehicle that is closer (resp. farther) from the conflict point.

For the CBF-based approach, because the ego vehicle cannot control the surrounding vehicles' behaviors, we use a variant of the centralized CBF -- the predictor-corrector collision avoidance (PCCA) of  \cite{pcca}.  Note that differently from \cite{pcca}, the ego vehicle's action space is bounded in the driving scenario. If the generated action from the CBF exceeds the bound, then the bound value is applied in our implementation.

 \begin{table*}[!h]
\centering
{
\caption{Comparative results: Potential game vs. Reinforcement learning vs. Control barrier function}
\begin{tabular}{c|c|c|c|c|c|c}
\hline  
\multicolumn{1}{c|}{\multirow{2}{*}{}}
&\multicolumn{2}{c|}{\multirow{2}{*}{Potential game}}
&\multicolumn{2}{c|}{\multirow{2}{*}{Control barrier function}}
&\multicolumn{2}{c}{\multirow{2}{*}{Reinforcement learning}}\\
\multicolumn{1}{c|}{\multirow{2}{*}{}}
&\multicolumn{2}{c|}{\multirow{2}{*}{}}
&\multicolumn{2}{c|}{\multirow{2}{*}{}}
&\multicolumn{2}{c}{\multirow{2}{*}{}}
\\
\hline
Surrounding vehicles' strategies & NE & Constant speed & NE & Constant speed & NE & Constant speed 
\\
\hline
\multirow{1}{*}{Collision rate} & $0/5000$ & $0/5000$ & $0/5000$ & $2229/5000$ & $0/5000$ & $2403/5000$
\\
\hline
\multirow{1}{*}{Average ego speed (m/s)} & $3.78$ & $3.05$ & $5.00$ & $4.19$ & $4.57$ & $4.89$
\\
\hline
\multirow{1}{*}{Avg/Max relative collision speed (m/s)} & N/A & N/A & N/A & 5.23/7.07 & N/A & $6.82/7.07$
\\
\hline
\multirow{1}{*}{Avg/Max ego speed at collision (m/s)} & N/A & N/A & N/A & 0.72/5.00 & N/A & $4.63/5.00$
\\
\hline
\multicolumn{1}{c|}{\multirow{1}{*}{Ave/Max computational time (s)}} &
\multicolumn{2}{c|}{\multirow{1}{*}{$0.05/0.10$}} &
\multicolumn{2}{c|}{\multirow{1}{*}{$<0.01/<0.01$}}&
\multicolumn{2}{c}{\multirow{1}{*}{$<0.01/<0.01$}} 
\\
\hline
\end{tabular}\label{table3}
}
\end{table*}

The statistical comparison of the above approaches is shown in Table \ref{table3}. (Performance in specific scenarios is shown in the last part at \textit{https://youtu.be/nQkpdQRcwEE}.) 
Table \ref{table3} leads to the following observations:
\begin{enumerate}
    \item The reinforcement learning approach performs well when the surrounding vehicles follow the NE strategies, while its performance  significantly degrades  when other vehicles follow the constant speed strategy. It is because the performance of RL is highly dependent on the training dataset. Since the deep Q-network is trained by IDM-governed surrounding vehicles, if a vehicle does not follow the IDM and is safety-agnostic, it can easily lead to a collision with the ego vehicle. In real world,  it is extremely hard, if not impossible, to reflect all situations and all vehicle behaviors  in the training dataset, and therefore, the robustness of the RL-based algorithm can be a major concern in autonomous driving. 
    \item The control barrier function based approach is not able to ensure safety either, when the surrounding vehicles follow the constant speed strategy. It is because to guarantee the forward invariance of a pre-defined safety set, the CBF based approach generally requires all agents' adherence to the CBF constraints \cite{CBF}. In a decentralized setting where other agents may not adhere to the CBF constraints, to ensure recursive feasibility, an unbounded action space of the ego vehicle is usually required \cite{pcca}, which is unrealistic in driving applications. 
    \item In contrast to RL and CBF, the potential game approach demonstrates significantly better robustness: It successfully avoids all collisions even if  the surrounding vehicles follow the constant speed strategy. 
    \item All three approaches are computationally feasible, as the running time is always less then the sampling time in all tested scenarios. 
\end{enumerate}
} 

\section{Conclusion}\label{VI}
This paper developed two potential game based frameworks to {formulate the process of} decision-making in autonomous driving when multiple traffic agents {are} involved. Theoretical guarantees for the existence of NE and convergence of NE seeking algorithms {were presented}. Scalability challenge was addressed. Approaches to constructing potential games {for} autonomous driving applications {were} provided.  By  comprehensive numerical studies, we showed that our developed potential game frameworks can generate safe and effective decisions {for the ego vehicle} in diverse  traffic scenarios including lane-changing and intersection-crossing. {Comparative studies with the RL and the CBF based approaches showed that the developed potential game approach leads to the best robustness against safety-agnostic surrounding vehicles. }
Future work will include verification and validation of the proposed potential game approach in various traffic situations, and development of prediction and learning modules to account for uncertain intentions and/or unknown cost functions of other vehicles.

\bibliography{references}{}

\begin{thebibliography}{10}
\providecommand{\url}[1]{#1}
\csname url@samestyle\endcsname
\providecommand{\newblock}{\relax}
\providecommand{\bibinfo}[2]{#2}
\providecommand{\BIBentrySTDinterwordspacing}{\spaceskip=0pt\relax}
\providecommand{\BIBentryALTinterwordstretchfactor}{4}
\providecommand{\BIBentryALTinterwordspacing}{\spaceskip=\fontdimen2\font plus
\BIBentryALTinterwordstretchfactor\fontdimen3\font minus
  \fontdimen4\font\relax}
\providecommand{\BIBforeignlanguage}[2]{{%
\expandafter\ifx\csname l@#1\endcsname\relax
\typeout{** WARNING: IEEEtran.bst: No hyphenation pattern has been}%
\typeout{** loaded for the language `#1'. Using the pattern for}%
\typeout{** the default language instead.}%
\else
\language=\csname l@#1\endcsname
\fi
#2}}
\providecommand{\BIBdecl}{\relax}
\BIBdecl

\bibitem{online1}
BuyShares, ``Global autonomous car market to grow by 36$\%$ and hit a $\$37 b$
  value by 2023,''
  https://buyshares.co.uk/global-autonomous-car-market-to-grow-by-36-and-hit-a-37b-value-by-2023/,
  Tech. Rep., 2021.

\bibitem{challenge3}
M.~Campbell, M.~Egerstedt, J.~P. How, and R.~M. Murray, ``Autonomous driving in
  urban environments: approaches, lessons and challenges,'' \emph{Philosophical
  Transactions of the Royal Society A: Mathematical, Physical and Engineering
  Sciences}, vol. 368, no. 1928, pp. 4649--4672, 2010.

\bibitem{challenge1}
D.~Gonz{\'a}lez, J.~P{\'e}rez, V.~Milan{\'e}s, and F.~Nashashibi, ``A review of
  motion planning techniques for automated vehicles,'' \emph{IEEE Transactions
  on Intelligent Transportation Systems}, vol.~17, no.~4, pp. 1135--1145, 2015.

\bibitem{challenge2}
B.~R. Kiran, I.~Sobh, V.~Talpaert, P.~Mannion, A.~A. Al~Sallab, S.~Yogamani,
  and P.~P{\'e}rez, ``Deep reinforcement learning for autonomous driving: A
  survey,'' \emph{IEEE Transactions on Intelligent Transportation Systems},
  2021.

\bibitem{decision_add1}
S.~Noh, ``Decision-making framework for autonomous driving at road
  intersections: Safeguarding against collision, overly conservative behavior,
  and violation vehicles,'' \emph{IEEE Transactions on Industrial Electronics},
  vol.~66, no.~4, pp. 3275--3286, 2018.

\bibitem{decision2}
C.-J. Hoel, K.~Driggs-Campbell, K.~Wolff, L.~Laine, and M.~J. Kochenderfer,
  ``Combining planning and deep reinforcement learning in tactical decision
  making for autonomous driving,'' \emph{IEEE Transactions on Intelligent
  Vehicles}, vol.~5, no.~2, pp. 294--305, 2019.

\bibitem{decision3}
P.~Hang, C.~Lv, Y.~Xing, C.~Huang, and Z.~Hu, ``Human-like decision making for
  autonomous driving: A noncooperative game theoretic approach,'' \emph{IEEE
  Transactions on Intelligent Transportation Systems}, 2020.

\bibitem{interaction_2}
A.~Rasouli and J.~K. Tsotsos, ``Autonomous vehicles that interact with
  pedestrians: A survey of theory and practice,'' \emph{IEEE Transactions on
  Intelligent Transportation Systems}, vol.~21, no.~3, pp. 900--918, 2019.

\bibitem{interaction_1}
L.~Hou, L.~Xin, S.~E. Li, B.~Cheng, and W.~Wang, ``Interactive trajectory
  prediction of surrounding road users for autonomous driving using
  structural-lstm network,'' \emph{IEEE Transactions on Intelligent
  Transportation Systems}, vol.~21, no.~11, pp. 4615--4625, 2019.

\bibitem{interaction_3}
C.~Xu, W.~Zhao, L.~Li, Q.~Chen, D.~Kuang, and J.~Zhou, ``A nash q-learning
  based motion decision algorithm with considering interaction to traffic
  participants,'' \emph{IEEE Transactions on Vehicular Technology}, vol.~69,
  no.~11, pp. 12\,621--12\,634, 2020.

\bibitem{nan_game}
N.~Li, D.~W. Oyler, M.~Zhang, Y.~Yildiz, I.~Kolmanovsky, and A.~R. Girard,
  ``Game theoretic modeling of driver and vehicle interactions for verification
  and validation of autonomous vehicle control systems,'' \emph{IEEE
  Transactions on Control Systems Technology}, vol.~26, no.~5, pp. 1782--1797,
  2017.

\bibitem{game_sta}
Q.~Zhang, R.~Langari, H.~E. Tseng, D.~Filev, S.~Szwabowski, and S.~Coskun, ``A
  game theoretic model predictive controller with aggressiveness estimation for
  mandatory lane change,'' \emph{IEEE Transactions on Intelligent Vehicles},
  vol.~5, no.~1, pp. 75--89, 2019.

\bibitem{Victor}
V.~Lopez, F.~Lewis, M.~Liu, Y.~Wan, S.~Nageshrao, and D.~Filev,
  ``Game-theoretic lane-changing decision making and payoff learning for
  autonomous vehicles,'' \emph{IEEE Transactions on Vehicular Technology},
  2022.

\bibitem{mine_1}
M.~Liu, Y.~Wan, F.~Lewis, S.~Nageshrao, and D.~Filev, ``A three-level
  game-theoretic decision-making framework for autonomous vehicles,''
  \emph{IEEE Transactions on Intelligent Transportation Systems}, 2022.

\bibitem{game_new_1}
F.~Camara, N.~Bellotto, S.~Cosar, F.~Weber, D.~Nathanael, M.~Althoff, J.~Wu,
  J.~Ruenz, A.~Dietrich, G.~Markkula \emph{et~al.}, ``Pedestrian models for
  autonomous driving part ii: high-level models of human behavior,'' \emph{IEEE
  Transactions on Intelligent Transportation Systems}, 2020.

\bibitem{game_merge}
H.~Kita, ``A merging--giveway interaction model of cars in a merging section: a
  game theoretic analysis,'' \emph{Transportation Research Part A: Policy and
  Practice}, vol.~33, no. 3-4, pp. 305--312, 1999.

\bibitem{game_leader}
N.~Li, Y.~Yao, I.~Kolmanovsky, E.~Atkins, and A.~R. Girard, ``Game-theoretic
  modeling of multi-vehicle interactions at uncontrolled intersections,''
  \emph{IEEE Transactions on Intelligent Transportation Systems}, 2020.

\bibitem{suzhou}
Q.~Dai, X.~Xu, W.~Guo, S.~Huang, and D.~Filev, ``Towards a systematic
  computational framework for modeling multi-agent decision-making at micro
  level for smart vehicles in a smart world,'' \emph{Robotics and Autonomous
  Systems}, vol. 144, p. 103859, 2021.

\bibitem{first}
D.~Monderer and L.~S. Shapley, ``Potential games,'' \emph{Games and economic
  behavior}, vol.~14, no.~1, pp. 124--143, 1996.

\bibitem{economics}
D.~Cai, S.~Bose, and A.~Wierman, ``On the role of a market maker in networked
  cournot competition,'' \emph{Mathematics of Operations Research}, vol.~44,
  no.~3, pp. 1122--1144, 2019.

\bibitem{nature}
J.~Wu and D.~Zusai, ``A potential game approach to modelling evolution in a
  connected society,'' \emph{Nature human behaviour}, vol.~3, no.~6, pp.
  604--610, 2019.

\bibitem{wirless}
K.~Yamamoto, ``A comprehensive survey of potential game approaches to wireless
  networks,'' \emph{IEICE Transactions on Communications}, vol.~98, no.~9, pp.
  1804--1823, 2015.

\bibitem{hard}
Y.~Hino, ``An improved algorithm for detecting potential games,''
  \emph{International Journal of Game Theory}, vol.~40, no.~1, pp. 199--205,
  2011.

\bibitem{game_book}
Y.~Shoham and K.~Leyton-Brown, \emph{Multiagent systems: Algorithmic,
  game-theoretic, and logical foundations}.\hskip 1em plus 0.5em minus
  0.4em\relax Cambridge University Press, 2008.

\bibitem{add_PSNE}
H.~Lu, ``On the existence of pure-strategy nash equilibrium,'' \emph{Economics
  Letters}, vol.~94, no.~3, pp. 459--462, 2007.

\bibitem{add_exponential}
S.~Durand and B.~Gaujal, ``Complexity and optimality of the best response
  algorithm in random potential games,'' in \emph{International Symposium on
  Algorithmic Game Theory}, 2016, pp. 40--51.

\bibitem{potential_book}
Q.~D. L{\~a}, Y.~H. Chew, and B.-H. Soong, \emph{Potential Game Theory}.\hskip
  1em plus 0.5em minus 0.4em\relax Springer, 2016.

\bibitem{add_jerk}
M.~Elbanhawi, M.~Simic, and R.~Jazar, ``In the passenger seat: investigating
  ride comfort measures in autonomous cars,'' \emph{IEEE Intelligent
  transportation systems magazine}, vol.~7, no.~3, pp. 4--17, 2015.

\bibitem{add_travelefficiency}
M.~Zhou, Y.~Yu, and X.~Qu, ``Development of an efficient driving strategy for
  connected and automated vehicles at signalized intersections: A reinforcement
  learning approach,'' \emph{IEEE Transactions on Intelligent Transportation
  Systems}, vol.~21, no.~1, pp. 433--443, 2019.

\bibitem{add_fuel}
L.~Zhang, K.~Peng, X.~Zhao, and A.~J. Khattak, ``New fuel consumption model
  considering vehicular speed, acceleration, and jerk,'' \emph{Journal of
  Intelligent Transportation Systems}, pp. 1--13, 2021.

\bibitem{payoff1}
Q.~Dai, D.~Shen, J.~Wang, S.~Huang, and D.~Filev, ``Calibration of human
  driving behavior and preference using naturalistic traffic data,''
  \emph{arXiv preprint arXiv:2105.01820}, 2021.

\bibitem{payoff2}
C.~Hubmann, M.~Becker, D.~Althoff, D.~Lenz, and C.~Stiller, ``Decision making
  for autonomous driving considering interaction and uncertain prediction of
  surrounding vehicles,'' in \emph{IEEE Intelligent Vehicles Symposium}, 2017,
  pp. 1671--1678.

\bibitem{payoff3}
J.~F. Fisac, E.~Bronstein, E.~Stefansson, D.~Sadigh, S.~S. Sastry, and A.~D.
  Dragan, ``Hierarchical game-theoretic planning for autonomous vehicles,'' in
  \emph{International Conference on Robotics and Automation (ICRA)}, 2019, pp.
  9590--9596.

\bibitem{dynamics}
R.~Tian, N.~Li, I.~Kolmanovsky, Y.~Yildiz, and A.~R. Girard, ``Game-theoretic
  modeling of traffic in unsignalized intersection network for autonomous
  vehicle control verification and validation,'' \emph{IEEE Transactions on
  Intelligent Transportation Systems}, 2020.

\bibitem{bicycle_2}
J.~Kong, M.~Pfeiffer, G.~Schildbach, and F.~Borrelli, ``Kinematic and dynamic
  vehicle models for autonomous driving control design,'' in \emph{IEEE
  intelligent vehicles symposium (IV)}, 2015, pp. 1094--1099.

\bibitem{bicycle_validation}
P.~Polack, F.~Altch{\'e}, B.~d'Andr{\'e}a Novel, and A.~de~La~Fortelle, ``The
  kinematic bicycle model: A consistent model for planning feasible
  trajectories for autonomous vehicles?'' in \emph{IEEE intelligent vehicles
  symposium (IV)}, 2017, pp. 812--818.

\bibitem{add_T}
L.~Yang, X.~Li, W.~Guan, H.~M. Zhang, and L.~Fan, ``Effect of traffic density
  on drivers’ lane change and overtaking maneuvers in freeway situation—a
  driving simulator--based study,'' \emph{Traffic injury prevention}, vol.~19,
  no.~6, pp. 594--600, 2018.

\bibitem{add_deltaT}
X.~Cao, W.~Young, and M.~Sarvi, ``Exploring duration of lane change
  execution,'' in \emph{Australasian Transport Research Forum}, 2013.

\bibitem{add_collect}
J.~Xiao, Z.~Xiao, D.~Wang, V.~Havyarimana, C.~Liu, C.~Zou, and D.~Wu, ``Vehicle
  trajectory interpolation based on ensemble transfer regression,'' \emph{IEEE
  Transactions on Intelligent Transportation Systems}, 2021.

\bibitem{inverse_RL}
J.~Liu, L.~N. Boyle, and A.~Banerjee, ``An inverse reinforcement learning
  approach for customizing automated lane change systems,'' \emph{IEEE
  Transactions on Vehicular Technology}, 2022.

\bibitem{ga}
MathWorks, ``Find minimum of function using genetic algorithm,''
  \url{https://www.mathworks.com/help/gads/ga.html}.

\bibitem{add_ga}
D.~I. Arkhipov, D.~Wu, T.~Wu, and A.~C. Regan, ``A parallel genetic algorithm
  framework for transportation planning and logistics management,'' \emph{IEEE
  Access}, vol.~8, pp. 106\,506--106\,515, 2020.

\bibitem{add_DQN_2}
H.~Van~Hasselt, A.~Guez, and D.~Silver, ``Deep reinforcement learning with
  double q-learning,'' in \emph{Proceedings of the AAAI conference on
  artificial intelligence}, vol.~30, no.~1, 2016.

\bibitem{add_DQN_1}
S.~Nageshrao, H.~E. Tseng, and D.~Filev, ``Autonomous highway driving using
  deep reinforcement learning,'' in \emph{2019 IEEE International Conference on
  Systems, Man and Cybernetics (SMC)}, 2019, pp. 2326--2331.

\bibitem{IDM}
M.~Treiber, A.~Hennecke, and D.~Helbing, ``Congested traffic states in
  empirical observations and microscopic simulations,'' \emph{Physical review
  E}, vol.~62, no.~2, p. 1805, 2000.

\bibitem{pcca}
M.~Santillo and M.~Jankovic, ``Collision free navigation with interacting,
  non-communicating obstacles,'' in \emph{American Control Conference (ACC)},
  2021, pp. 1637--1643.

\bibitem{CBF}
A.~D. Ames, X.~Xu, J.~W. Grizzle, and P.~Tabuada, ``Control barrier function
  based quadratic programs for safety critical systems,'' \emph{IEEE
  Transactions on Automatic Control}, vol.~62, no.~8, pp. 3861--3876, 2016.

\bibitem{gradient_theorem}
R.~Williamson and H.~Trotter, \emph{Multivariable Mathematics}.\hskip 1em plus
  0.5em minus 0.4em\relax Pearson Education, 2004.

\end{thebibliography}
\bibliographystyle{IEEEtran}

\section*{Appendix}
\subsection{Proof of Proposition \ref{p1}}
\emph{Necessity}. We first prove that if \eqref{Potential_def} holds, then \eqref{defi_c} holds.   Let $\mathbf{r}$ be a unit vector in $\mathbb{R}^{m_i}$. If \eqref{Potential_def} holds, then the following equation holds,
\begin{equation}\label{23}
\begin{split}
   & \lim_{h\to 0}\frac{V_i^t(\mathbf{u}_i+h\mathbf{r},\mathbf{u}_{-i})-V_i^t(\mathbf{u}_i,\mathbf{u}_{-i})}{h}\\
   &=\lim\limits_{h\to 0}\frac{F^t(\mathbf{u}_i+h\mathbf{r},\mathbf{u}_{-i})-F^t(\mathbf{u}_i,\mathbf{u}_{-i})}{h}, \quad \forall\mathbf{r}. 
    \end{split}
\end{equation}
 Because $\lim\limits_{h\to 0}\frac{V_i^t(\mathbf{u}_i+h\mathbf{r},\mathbf{u}_{-i})-V_i^t(\mathbf{u}_i,\mathbf{u}_{-i})}{h}= \frac{\partial V_i^t(\mathbf{u}_i,\mathbf{u}_{-i})}{\partial \mathbf{u}_i}\cdot \mathbf{r}$ and  $\lim\limits_{h\to 0}\frac{F^t(\mathbf{u}_i+h\mathbf{r},\mathbf{u}_{-i})-F^t(\mathbf{u}_i,\mathbf{u}_{-i})}{h}= \frac{\partial F^t(\mathbf{u}_i,\mathbf{u}_{-i})}{\partial \mathbf{u}_i}\cdot \mathbf{r}$, Equation \eqref{23} leads to \eqref{defi_c}.

\emph{Sufficiency}. Now we prove that if \eqref{defi_c} holds, then \eqref{Potential_def} holds. Because $\mathcal{S}_i$ is a connected set, for any given $\mathbf{u}_i\in \mathcal{S}_i$ and $\mathbf{u}'_i\in \mathcal{S}_i$, there must exist a continuous curve ${\gamma}$ that starts at $\mathbf{u}_i$ and ends at $\mathbf{u}'_i$. As such,  \eqref{defi_c} leads to
\begin{equation}\label{21}
    \int_{\gamma} \frac{\partial V_i^t(\Tilde{\mathbf{u}}_i,\mathbf{u}_{-i})}{\partial \Tilde{\mathbf{u}}_i}d\Tilde{\mathbf{u}}_i=\int_{\gamma} \frac{\partial F^t(\Tilde{\mathbf{u}}_i,\mathbf{u}_{-i})}{\partial \Tilde{\mathbf{u}}_i}d\Tilde{\mathbf{u}}_i,
\end{equation}
where $\int_{\gamma}$ represents the line integral along the curve $\gamma$.
According to the gradient theorem for line integrals \cite{gradient_theorem},  $\int_{\gamma} \frac{\partial V_i^t(\Tilde{\mathbf{u}}_i,\mathbf{u}_{-i})}{\partial \Tilde{\mathbf{u}}_i}d\Tilde{\mathbf{u}}_i=V_i^t(\mathbf{u}_i',\mathbf{u}_{-i})-V_i^t(\mathbf{u}_i,\mathbf{u}_{-i})$, and $\int_{\gamma} \frac{\partial F^t(\Tilde{\mathbf{u}}_i,\mathbf{u}_{-i})}{\partial \Tilde{\mathbf{u}}_i}d\Tilde{\mathbf{u}}_i=F^t(\mathbf{u}_i',\mathbf{u}_{-i})-F^t(\mathbf{u}_i,\mathbf{u}_{-i})$. As such, Equation \eqref{21} leads to \eqref{Potential_def}.

\subsection{Proof of Lemma \ref{l4}}
{Let $(\mathbf{u}^*_i,\mathbf{u}^*_{-i})$ be a NE of the game \eqref{value}, i.e., $\mathbf{u}^*_i\in\argmin V_i^t(\mathbf{u}_i,\mathbf{u}^*_{-i})$  $\forall i\in\mathcal{N}$. From \eqref{defi_c}, it follows that $V_i^t(\mathbf{u}_i,\mathbf{u}_{-i})-F^t(\mathbf{u}_i,\mathbf{u}_{-i})=C(\mathbf{u}_{-i})$ for some function $C(\mathbf{u}_{-i})$, and hence  $\mathbf{u}^*_i\in\argmin F^t(\mathbf{u}_i,\mathbf{u}^*_{-i})$    $\forall i\in\mathcal{N}$, i.e., $(\mathbf{u}^*_i,\mathbf{u}^*_{-i})$ is a also NE for the identical-interest game with all agents' cumulative costs equal to $F^t$. On the other hand, if $(\mathbf{u}^*_i,\mathbf{u}^*_{-i})$ is a NE for the identical-interest game, then similar arguments imply that it is also a NE of the game \eqref{value}.}

\subsection{Proof of Lemma \ref{l5}}
For continuous potential games with compact $\mathcal{S}$ {and continuous $F^t$, a minimum of $F^t$ as a function of $\mathbf{u}$ always exists, according to Weierstrass theorem \cite{gradient_theorem}}. Hence, a NE exists  according to {Lemma \ref{l4}}. If $F^t$ is strictly convex, then it has a unique minimum, and as such, the game has a unique NE.

\begin{IEEEbiography}[{\includegraphics[width=1in,height=1.25in,clip,keepaspectratio]{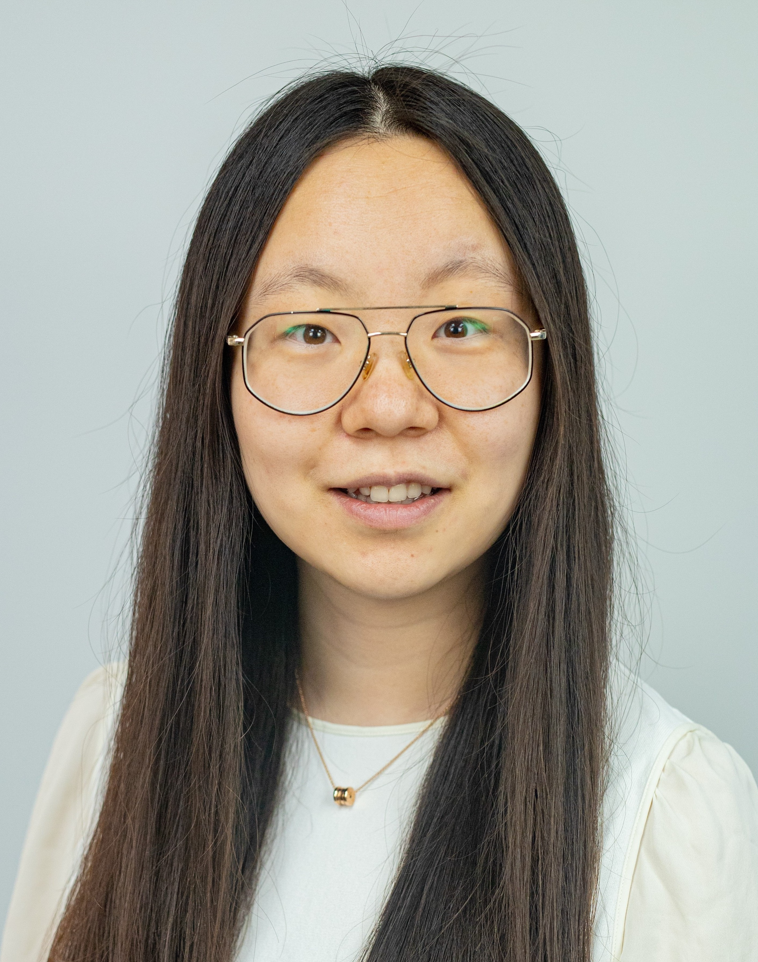}}]{Mushuang Liu} is an Assistant Professor in the Department of Mechanical and Aerospace Engineering at the University of Missouri, Columbia, MO. She worked as a Postdoc in the Department of Aerospace Engineering at the University of Michigan, Ann Arbor, MI, 2021-2022. She received her Ph.D degree from the University of Texas at Arlington in 2020 and her B.S. degree from the University of Electronic Science and Technology of China in 2016. Her research lies in  control and learning for multi-agent systems.
\end{IEEEbiography}

\begin{IEEEbiography}[{\includegraphics[width=1in,height=1.25in,clip,keepaspectratio]{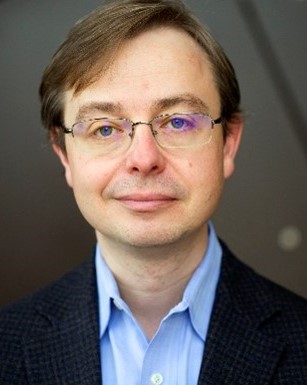}}]{Ilya Kolmanovsky}  is a professor in the department of aerospace engineering at the University of Michigan, Ann Arbor, MI, USA, with research interests in control theory for systems with state and control constraints, and in control applications to aerospace and automotive systems. He received his Ph.D. degree in aerospace engineering from the University of Michigan in 1995.  Prior to joining the University of Michigan as a faculty in 2010, Kolmanovsky was with Ford Research and Advanced Engineering in Dearborn, Michigan for close to 15 years. He is a Fellow of IEEE, IFAC and NAI, and a Senior Editor of IEEE Transactions on Control Systems Technology.
\end{IEEEbiography}

\begin{IEEEbiography}[{\includegraphics[width=1in,height=1.25in,clip,keepaspectratio]{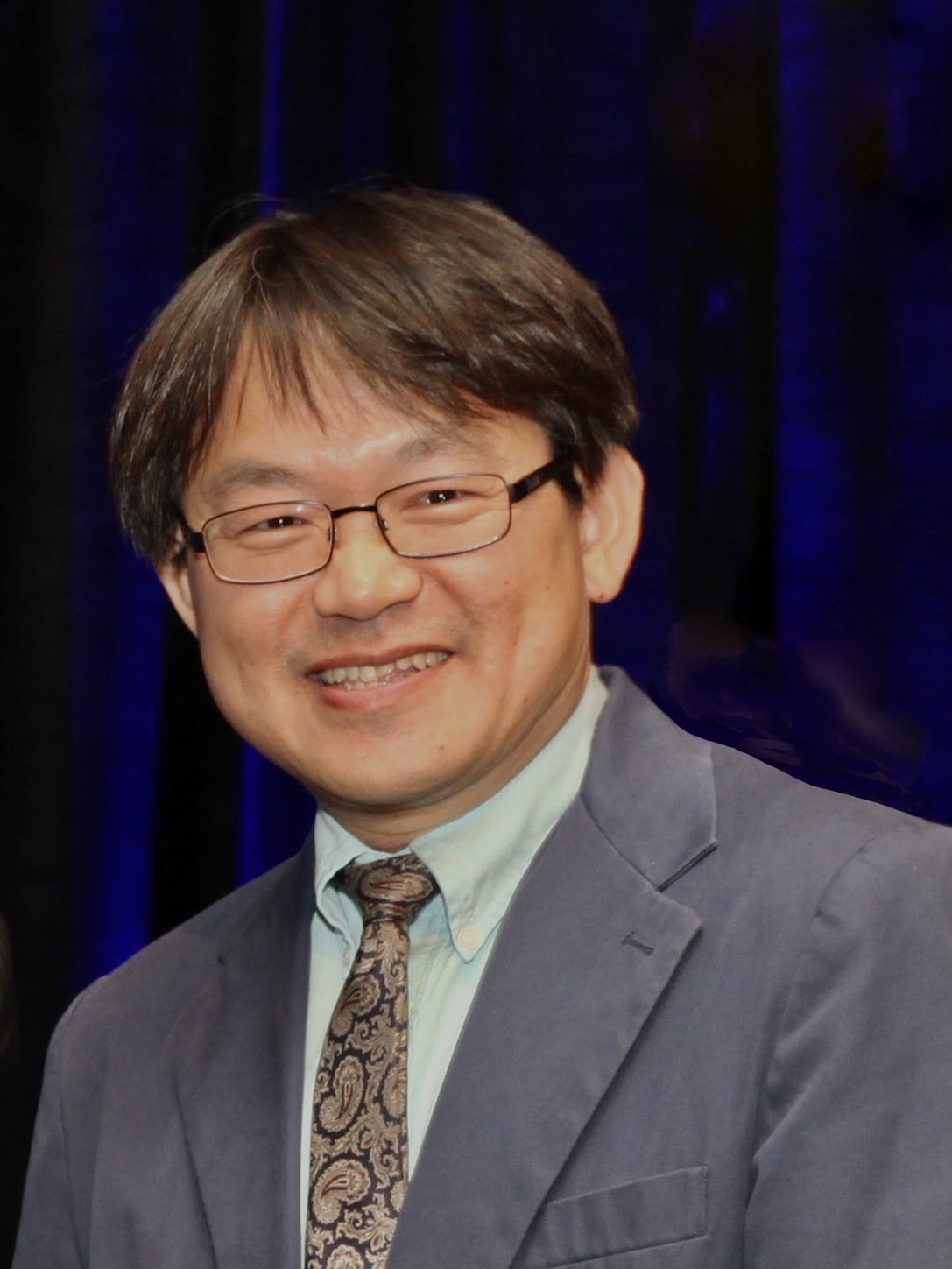}}]{H. Eric Tseng}  received the B.S. degree from the National Taiwan University, Taipei, Taiwan, in 1986, and the M.S. and Ph.D. degrees in mechanical engineering from the University of California at Berkeley, Berkeley, in 1991 and 1994, respectively. In 1994, he joined Ford Motor Company. At Ford, he is currently a Senior Technical Leader of Controls and Automated Systems in Research and Advanced Engineering. Many of his contributed technologies led to production vehicles implementation. His technical achievements have been recognized internally seven times with Ford’s highest technical award—the Henry Ford Technology Award, as well as externally by the American Automatic Control Council with Control Engineering Practice Award in 2013. He has over 100 U.S. patents and over 120 publications. He is an NAE Member.
\end{IEEEbiography}

\begin{IEEEbiography}[{\includegraphics[width=1in,height=1.25in,clip,keepaspectratio]{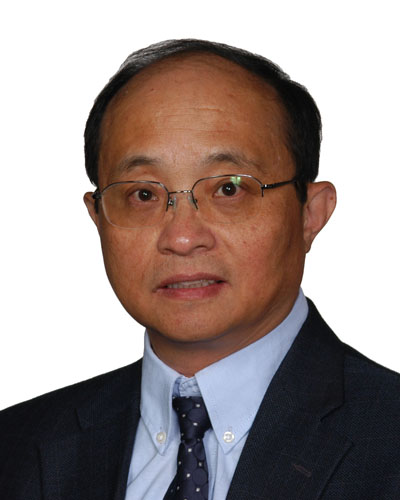}}]{Suzhou Huang}  is an executive technical leader in Research and Advanced Engineering at Ford Motor Company. His research interests now mostly concentrate on autonomous and connected vehicle technologies, human driving behavior modeling, and smart transportation systems using methodologies of ML/AI and game theory. Prior to his current position, he was director, Analytics R$\&$D in Global Data Insights and Analytics at Ford Motor and at Ford Credit, developing advanced analytics models in a variety of areas, such as production planning, inventory management, pricing, marketing program evaluation, auto financing, risk management, capital allocation, and regulation compliance. Suzhou Huang earned a Ph.D. in theoretical physics from MIT. Since he joined Ford in 1998 he transformed himself into an applied micro-economist/econometrician.
\end{IEEEbiography}

\begin{IEEEbiography}[{\includegraphics[width=1in,height=1.5in,clip,keepaspectratio]{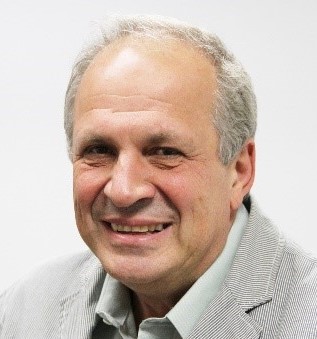}}]{Dimitar Filev} is Senior Henry Ford Technical Fellow in Control and AI with Research $\&$ Advanced Engineering – Ford Motor Company. His research is in computational intelligence, AI and intelligent control, and their applications to autonomous driving, vehicle systems, and automotive engineering.  He holds over 100 granted US patents and has been awarded with the IEEE SMCS 2008 Norbert Wiener Award and the 2015 Computational Intelligence Pioneer’s Award. Dr. Filev is a Fellow of the IEEE and a member of the National Academy of Engineering. He was President of the IEEE Systems, Man, $\&$ Cybernetics Society (2016-2017).
\end{IEEEbiography}

\begin{IEEEbiography}[{\includegraphics[width=1in,height=1.5in,clip,keepaspectratio]{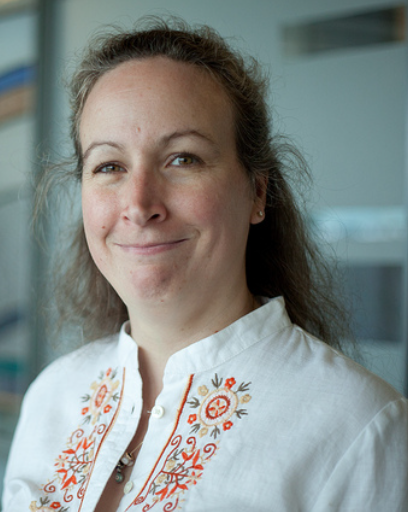}}]{Anouck Girard} received
the Ph.D. degree in ocean engineering from the
University of California at Berkeley, Berkeley, CA, USA, in 2002. She has been with the University of Michigan, Ann Arbor, MI, USA, since 2006, where she is currently a Professor of aerospace engineering. She has coauthored the book Fundamentals of Aerospace Navigation and Guidance (Cambridge University Press, 2014). Her current research interests include vehicle dynamics and control systems. She was a recipient of the Silver Shaft Teaching Award from the University of Michigan and the Best Student Paper Award from the American Society of Mechanical Engineers.
\end{IEEEbiography}

\end{document}